\documentclass[superscriptaddress,nofootinbib,tightenlines,preprint]{revtex4-1}
\pdfoutput=1
\usepackage[utf8]{inputenc}

\usepackage[pdftex,breaklinks,colorlinks,linktocpage,
citecolor=blue,
urlcolor=blue]{hyperref}

\usepackage{bm}
\usepackage{amssymb}
\usepackage{amsfonts}
\usepackage{mathrsfs}
\usepackage{amsmath}
\usepackage{amsthm}
\usepackage{xcolor}
\usepackage{bbold}
\usepackage{braket}
\usepackage{physics}
\usepackage{graphicx}
\usepackage{tensor}

\newcommand{\be}{\begin{equation}}
\newcommand{\ee}{\end{equation}}
\newcommand{\bea}{\begin{eqnarray}}
\newcommand{\eea}{\end{eqnarray}}
\def\bml{\begin{subequations}}
\def\blea{\bml\begin{eqnarray}}
\def\eml{\end{subequations}}
\def\elea{\end{eqnarray}\eml}

\theoremstyle{definition}
\newtheorem{theorem}{Theorem}[section]
\newtheorem{prop}[theorem]{Proposition}
\newtheorem{rem}[theorem]{Remark}
\newtheorem{cor}[theorem]{Corollary}
\newtheorem{lem}[theorem]{Lemma}
\newtheorem{defn}[theorem]{Definition}

\newcommand{\x}{\mathbf{x}}

\def\fmax{\phi_{\text{max}}}
\newcommand{\nord}[1]{{:}#1{:}}

\DeclareMathOperator{\Ric}{Ric}

\DeclareMathOperator{\cut}{Cut}
\DeclareMathOperator{\R}{\mathbb{R}}
\DeclareMathOperator{\Regplus}{Reg^+_{\eta}}

\date{\today}

\begin{document}

\title{Hawking-type singularity theorems for worldvolume energy inequalities}

\author{Melanie Graf}
\email{graf@math.uni-tuebingen.de}
\affiliation{University of Tübingen, Faculty of Mathematics, Auf der Morgenstelle 10, 72076 Tübingen, Germany}

\author{Eleni-Alexandra Kontou}
\email{eleni.kontou@kcl.ac.uk}
\affiliation{ITFA and GRAPPA, Universiteit van Amsterdam, Science Park 904, Amsterdam, the Netherlands}
\affiliation{Department of Mathematics, King’s College London, Strand, London WC2R 2LS, United Kingdom}

\author{Argam Ohanyan}
\email{argam.ohanyan@univie.ac.at}
\affiliation{Department of Mathematics, University of Vienna, Oskar-Morgenstern Platz 1, 1090 Vienna, Austria}

\author{Yasmin Schinnerl}
\email{b.schinnerl@gmail.com}
\affiliation{Department of Mathematics, University of Vienna, Oskar-Morgenstern Platz 1, 1090 Vienna, Austria}

\begin{abstract}

The classical singularity theorems of R.~Penrose and S.~Hawking from the 1960s show that, given a pointwise energy condition (and some causality as well as initial assumptions), spacetimes cannot be geodesically complete. 
Despite their great success, the theorems leave room for physically relevant improvements, especially regarding the classical energy conditions as essentially any quantum field theory necessarily violates them. While singularity theorems with weakened energy conditions exist for worldline integral bounds, so called worldvolume bounds are in some cases more applicable than the worldline ones, such as the case of some massive free fields. In this paper we study integral Ricci curvature bounds based on worldvolume quantum strong energy inequalities. Under the additional assumption of a - potentially very negative - global timelike Ricci curvature bound, a Hawking type singularity theorem is proven. 
Finally, we apply the theorem to a cosmological scenario proving past geodesic incompleteness in cases where the worldline theorem was inconclusive.

\vskip 1em

\noindent
\emph{Keywords:} Hawking's singularity theorem, worldvolume energy conditions, quantum energy inequality

\medskip
\noindent
\emph{MSC2020:} 83C75, 
        53C50, 
       53B30, 
       70S20 
\end{abstract}

\maketitle

\tableofcontents

\newpage

\section{Introduction}
\label{sec:introduction}

The singularity theorems of Penrose \cite{Penrose:1964wq} and Hawking \cite{Hawking:1966sx} represent an important breakthrough in classical general relativity. They showed that gravitational systems will originate or terminate in a singularity, defined in that context as geodesic incompleteness, under general conditions. For a recent review article on singularity theorems, see \cite{senovilla2022critical}.

Both theorems, along with mild causality conditions, require the non-negativity of contractions of the Ricci tensor. These conditions are assumed to represent properties of the matter content of the spacetime since, with the use of the Einstein Equation, they can be re-written as energy conditions, i.e.\ restrictions on contractions of the stress-energy tensor. 
The Hawking singularity theorem uses the timelike convergence condition
\be
R_{\mu \nu}U^\mu U^\nu \geq 0 \,,
\ee
where $U^\mu$ is a timelike vector. With the use of the Einstein Equation it becomes the strong energy condition (SEC)
\be
\label{eqn:eed}
\rho_U \equiv \left(T_{\mu \nu} -\frac{T g_{\mu \nu}}{n-2} \right) U^\mu U^\nu \geq 0 \,,
\ee
where $n$ is the spacetime dimension and $T$ the trace of the stress-energy tensor. The quantity on the left is the effective energy density (EED), a term introduced in \cite{Brown:2018hym}. The SEC has an unclear physical interpretation \cite{Kontou:2020bta, Curiel:2014zba} and is easily violated by classical fields, for example the free massive Klein-Gordon field. The situation is even worse when including quantum effects. It has long been known \cite{Epstein:1965zza} that quantum fields can admit negative energies violating all pointwise energy conditions.
Average energy conditions and quantum energy inequalities (QEIs) are weaker restrictions on the energy density (and similar quantities) or its renormalized expectation value in the case of quantum fields. Starting from the seminal work of Ford \cite{Ford:1978qya} QEIs have been derived for flat and curved spacetimes and a variety of mainly free fields (see Refs.~\cite{Kontou:2020bta}, \cite{Fewster2017QEIs} and references within). They often take the form of 
\be
\label{eqn:qeigen}
\langle \rho(f) \rangle_\omega \geq -|||f||| \,, 
\ee
where $f$ is a smooth compactly supported real-valued function. Here $\rho$ is the renormalized energy density, $\omega$ is a Hadamard state, and $||| \cdot |||$ a Sobolev norm.
The integral is more often over a timelike geodesic but it can also be over a worldvolume. 

With the use of the classical or the semiclassical Einstein Equation average energy conditions and QEIs become conditions on the Ricci tensor and may be used as alternative assumptions for singularity theorems. Starting with the work of Tipler \cite{Tipler:1978zz} various authors have proven singularity theorems with weaker energy conditions \cite{Borde:1987qr, Roman:1988vv}. First Fewster and Galloway \cite{Fewster:2010gm} proved singularity theorems with a condition of the form of \eqref{eqn:qeigen} using the Raychaudhuri equation. More recently Fewster and Kontou \cite{Fewster:2019bjg} presented proofs with a similar energy condition using index form methods and estimated the required initial contraction or mean normal curvature of a Cauchy surface. 

All of these works use energy conditions averaged over a single timelike or null geodesic. However, there is strong motivation to prove singularity theorems using worldvolume bounds instead.
In the timelike case the worldvolume bounds are often stronger than the worldline ones for both classical \cite{Brown:2018hym} and quantum fields \cite{Fewster:2018pey}. In particular in the case of a massive scalar, the mass dependence can be removed in the case of worldvolume bounds.

More important motivation to consider worldvolume inequalities arises in the null case. Fewster and Roman showed with an explicit counterexample \cite{Fewster:2002ne} that QEIs over a null geodesic segment do not admit finite lower bounds. However, this is not the case if the integral is over two null directions as it was recently proven \cite{Fliss:2021phs}. Using such an inequality as an assumption to a singularity theorem would lead to the proof of a semiclassical singularity theorem for null geodesic incompleteness. Previously only two semi-classical singularity theorems have been proven; one for the timelike incompleteness case \cite{Fewster:2021mmz} and one for the null incompleteness case \cite{Freivogel:2020hiz} but the latter required the introduction of a theory dependent UV cutoff significantly weakening the result.

The consequence of the classical singularity theorems is the existence of an incomplete causal geodesic. If one could simply extend the spacetime to a complete one this would in some sense signify that the incompleteness were not physically reasonable as there would be another better suited spacetime to model the situation.
Of course, the singularity theorems also imply that one cannot accomplish this with a $C^2$-Lorentzian metric without violating either the energy or causality condition.

Therefore the question arises whether one could instead resolve this issue by extending with a lower regularity Lorentzian metric,
i.e.\ given a $C^2$-spacetime satisfying the assumptions of the singularity theorems, could there be a (timelike and/or null geodesically) complete low regularity extension of it, satisfying appropriate generalizations of the assumptions? In other words could the singularity predicted by the smooth theorems merely be due to a drop in regularity (which could be akin to a jump or delta distribution in the curvature)?

This question is essentially equivalent to whether singularity theorems continue to hold for lower regularity metrics which was already raised and extensively discussed in \cite[Sec.~8.4]{hawking1973large}. Nevertheless it has only recently become rigorously tractable mathematically (cf.~\cite{kunzinger2015hawking, kunzinger2015penrose, graf2018hawking, graf2020singularity, schinnerl2021note, kunzinger2022hawking}) and trying to resolve it fully continues to be an active pursuit of major physical importance because a plethora of physical models involve spacetime metrics which are not $C^2$ (e.g., the Oppenheimer–-Snyder model of a collapsing star, matched spacetimes, gravitational shock waves, thin mass shells, \dots). For a recent comprehensive review of this research area see \cite{steinbauer2022singularity}. 

The current proofs of low-regularity singularity theorems commonly employ an approximation based approach aiming to distill and exploit some stability properties of both the conditions in the singularity theorems and the formation or non-formation of singularities. Note however that already within the class of smooth metrics, Ricci curvature being non-negative along timelike or null geodesics is unstable: Even if a metric $g_0$ satisfies the timelike or null convergence condition, metrics $g_\varepsilon $ converging to $g_0$ will in general violate these conditions by some small error, with the exact nature of the introduced error depending on the topology of the convergence. If $g_0$ is a non-smooth metric satisfying only a distributional version of the strong or null energy condition, then the possible topologies of convergence and therefore the form of the energy condition recovered for $g_\varepsilon$ are limited by the initial regularity of $g_0$. The current state of the art regarding low-regularity singularity theorems \cite{graf2020singularity, kunzinger2022hawking} considers $g_0\in C^{1}$. 
It was shown that in this case - for a family of approximations $g_\varepsilon$ constructed in a very specific way -  convergence is strong enough, but only barely, to recover a classical energy condition for the approximating $g_\varepsilon $ up to  ``pointwise small'' violations. That is one obtains a condition of the form $\mathrm{Ric}[g_\varepsilon ](X,X) \geq -\delta_\varepsilon$ for all $g_\varepsilon$-timelike/null vectors $X$ contained in some compact subset $K\subseteq TM$, where $0<\delta_\varepsilon\equiv \delta_\varepsilon (K) \to 0$. Such an energy condition is still amenable to classical Riccati or index form techniques to show that (for small enough $\delta_\varepsilon$ and large enough $K$ depending on the given initial condition) certain timelike geodesics must encounter focal points before a fixed time which then can be leveraged to obtain that, in the limit, a singularity must exist for the original metric.
However, if one were only to obtain weaker convergence
one expects that such pointwise estimates cannot be recovered from the distributional energy conditions placed on $g_0$  and that one instead will have to make do with integral bounds of some form, such as the $L^1$-bound we study in Theorem \ref{thm: singthmbasedonsegmineq}.\\

In this work we will address the issues arising from starting with a worldvolume bound by translating these bounds into bounds along single timelike geodesics by using ideas inspired
by the so-called segment inequality from Riemannian geometry due to Cheeger--Colding
\cite{cheeger1996lower} as well as an $L^1$-Ricci bound based Myers theorem by Sprouse \cite{sprouse2000integral}. More specifically, we consider volume integral bounds of the form
\begin{align}
       \frac{1}{\sigma(B)} \int_{\Omega_T^+(B)} |\Ric(U_p,U_p)_-| \, dvol_g(p) \leq C   
\end{align}
for suitable subsets $B$ of a Cauchy surface $\Sigma$ in $M$ (here $C$ is a constant, $U_p$ is the vector field that is tangent to normal geodesics from $\Sigma$ at $p$, $\sigma$ is the intrinsic Riemannian volume measure of $\Sigma$, and $\Omega_T^+(B)$ are those points that can be reached from $B$ when evolving until time $T$), as well as
\begin{align}
    \int_M \Ric(U_p,U_p) F(p)^2 \, dvol_g(p) \geq -Q_1 \|F\|_{L^2}^2 - Q_2 \|\dot{F}\|_{L^2}^2\quad\quad \forall F\in C^\infty_c(\mathcal{M}^+),
\end{align}
where $Q_1,Q_2$ are suitable constants and $\mathcal{M}^+$ is the cut point free future evolution of $\Sigma$, see Sec.~\ref{sec: geometricbackground} for definitions and Sec.~\ref{sec: singularitytheorems} for precise statements of our results. We then reduce these inequalities to inequalities along single timelike geodesics:
 By standard results from Lorentzian geometry the normal exponential map provides a diffeomorphism from a neighbourhood of the normal bundle of $\Sigma$ to a neighbourhood of $\Sigma$.  
 In such a region the evolution of the surface is well defined and one can use its geometric properties (like the infinitesimal area element and the mean curvature). However, as soon as the normal geodesics to $\Sigma$ encounter cut points, the evolution surfaces will degenerate or form caustics. 
Hence we will have to restrict our analysis to regions where $\Sigma$ has a well defined evolution. This will in turn allow us, using exponential coordinates, to naturally split a volume integral into an iterated integral along geodesics of the congruence and over a subset of $\Sigma$. However, the volume element will not be constant in time in these coordinates so the integrand in each line integral will not only involve the Ricci curvature along the geodesic but an additional geometric factor determined by the time evolution of the area element of $\Sigma$. To be able to deal with this additional factor, we need to develop new estimates on the time evolution of the area element from {\em below} resp.~{\em backwards in time}.
From there on we are able to extract bounds on a single line integral and use classical index form methods to deduce geodesic focusing given an appropriate initial mean curvature condition on $\Sigma$.\\

This paper is organized as follows. In Sec.~\ref{sec: geometricbackground} we give a series of definitions and known results on spacetime geometry that will be needed later. In Sec.~\ref{sec: forwardbackwardcomparison} we provide area comparison results that are used to derive bounds on integrals along geodesics from worldvolume integrals. In Sec.~\ref{sec: singularitytheorems} we present the main result of this paper, a singularity theorem with a worldvolume integrated energy condition as an assumption. In Sec.~\ref{sec: strongenergyinequality} we present a physical application of our theorem for the case of the classical non-minimally coupled scalar field. In Sec.~\ref{section: application} we apply the theorem in the case of a simple cosmological model proving past timelike geodesic incompleteness for a few different fields. We conclude in Sec.~\ref{sec:discussion} with a summary and discussion of future work. 

\subsection{Notation and Conventions}
\label{subsec: notationconventions}

Let us collect some remarks on the notation and conventions that we will use throughout. First, a spacetime is a pair $(M,g)$ where $M$ is a 2nd countable, Hausdorff smooth manifold, and $g$ is a (smooth) metric on $M$ with signature $(-,+,\dots,+)$. The D'Alembertian operator with respect to the metric $g$ is defined as $\Box_g := -g^{\mu \nu} \nabla_\mu \nabla_\nu$. Note that this sign convention for $\Box_g$ is consistent with the physics literature references of Sec.~\ref{sec: strongenergyinequality}-\ref{section: application}, but differs from the one used in most of the references for Sec.~\ref{sec: geometricbackground}-\ref{sec: singularitytheorems}. The Riemann curvature tensor of $g$ is $R(X,Y)Z = \nabla_X \nabla_Y Z - \nabla_Y \nabla_X Z - \nabla_{[X,Y]}Z$ and the Ricci tensor $\Ric$ is its $(1,3)$-contraction. For $A \subseteq M$, we write $\tau_A$ for the (forward) time separation to $A$, i.e.\ $\tau_A(p):=\sup_{q \in A} \tau(q,p)$. A Cauchy surface in a spacetime is a set which is uniquely met by each inextendible causal curve. Given a spacelike hypersurface $\Sigma$ of a spacetime $(M,g)$, the second fundamental form is the map $II:\mathfrak{X}(\Sigma) \times \mathfrak{X}(\Sigma) \to \mathfrak{X}(\Sigma)^{\perp}$, $II(X,Y):=(\nabla_X Y)^{\perp}$, where $\nabla$ is the Levi-Civita connection in $M$. For the mean curvature of a spacelike hypersurface, we adopt the conventions in \cite{treude2013volume}, i.e.\ if $\vec{n}$ is the future unit normal to the hypersurface, then the mean curvature is given by $H = \mathrm{tr}(\nabla \vec{n})$.\\
We use geometric units throughout this work (i.e.\ $G=c=1$) except for Section \ref{section: application}. Let us note that we will employ index notation in Section \ref{sec: strongenergyinequality}, as is customary in the physical literature.

\section{Geometric background}
\label{sec: geometricbackground}

In this section, we collect several well-known results from spacetime geometry which will be needed for what follows. Throughout, let $(M,g)$ be a smooth spacetime of dimension $n$ and let $\Sigma \subseteq M$ be a smooth, spacelike Cauchy surface with future timelike unit normal $\vec{n}$. We will often only consider future versions in definitions and results, past versions are obtained by duality.

As we will see in this section the exponential map not only yields a diffeomorphism from a
neighbourhood of its normal bundle to a neighbourhood of our spacelike Cauchy surface
$\Sigma$, it also provides us
with a foliation of that domain by 
hypersurfaces perpendicular to the
congruence of normal vector fields emanating from $\Sigma.$
This is guaranteed as long as the geodesic congruence does not
encounter cut points, as then the exponential map no longer is a diffeomorphism. We will mostly reference \cite{treude2011ricci} for proofs of the results presented in this section, for an alternative textbook reference see \cite[Sec.~9]{beem2017global}.

Since we will have to deal with the case where $M$ is not future geodesically complete, let $[0,s^+(\x))$ denote the maximal interval of definition for the unique geodesic $\gamma$ with initial data $\gamma(0)=\x \in \Sigma$, $\dot{\gamma}(0)=\vec{n}_\x$. 

Note that the future normal bundle of $\Sigma$ is of the form $N^+\Sigma =\{t\vec{n}_\x :  (t,\x) \in [0,\infty)\times \Sigma \} \cong [0,\infty) \times \Sigma$. Let $\mathcal{I}^+:=\{(t,\x) \in [0,\infty) \times \Sigma : t\in [0,s^+(\x)) \}\subseteq [0,\infty)\times \Sigma $.

\begin{defn}[Future normal exponential map]
\label{defn: normalexp}
The \textit{future normal exponential map} of $\Sigma$ is the map $\exp_{\Sigma}^+:\mathcal{I}^+ \to M$ defined by
\begin{align}
    \exp_{\Sigma}^+(t,\x):=\exp_{\x}(t \vec{n}_{\x}).
\end{align}
\end{defn}

Let us now recall the definition of (future) cut points
and the (future) cut function. 

\begin{defn}[Future cut function and future cut points]
\label{defn: cutfunctioncutpoints}
The \textit{future cut function} $c_{\Sigma}^+:\Sigma \to (0,\infty]$ of $\Sigma$ is defined via
\begin{align}
    c_{\Sigma}^+(\x):= \sup \{t \in [0,s^+(\x)) \mid \tau_{\Sigma}(\exp_{\x}(t \vec{n}_{\x})) = t\}.
\end{align}
If $c_{\Sigma}^+(\x) < s^+(\x)$ for some $\x\in \Sigma$, we call the point $\exp_{\x}(c_{\Sigma}^+(\x) \vec{n}_{\x})$ a \textit{future cut point} of $\Sigma$. We write $\cut^+(\Sigma)$ for the set of future cut points of $\Sigma$.
\end{defn}

\begin{lem}
\label{lem: cutfunctionlsc}
The future cut function $c_{\Sigma}^+$ is lower semicontinuous  and satisfies $c_{\Sigma}^+(\x) > 0$ for all $\x \in \Sigma$. Moreover, $\cut^+(\Sigma) \subseteq M$ is a set of measure zero with respect to the volume measure.
\end{lem}
\begin{proof}
Lower semicontinuity follows from \cite[Prop.\ 3.2.29]{treude2011ricci}. Since any geodesic normal to $\Sigma$ must maximize initially (cf.\ \cite[Cor.\ 3.2.23]{treude2011ricci}), it follows that $c_{\Sigma}^+(\x) > 0$ for all $\x \in \Sigma$. The claim about the measure of $\cut^+(\Sigma)$ is proven in \cite[Prop.\ 3.2.32]{treude2011ricci}. 
\end{proof}

\begin{lem}
\label{lem: normalexpdiffeo}
The future normal exponential map $\exp_{\Sigma}^+$ is a diffeomorphism (of manifolds with boundary) from a neighborhood $\mathcal{D}^+$ of $\{0\} \times \Sigma$ onto $\mathcal{M}^+:=J^+(\Sigma) \setminus \cut^+(\Sigma)$.
\end{lem}
\begin{proof}
This follows from \cite[Thm.\ 3.2.31]{treude2011ricci}.
\end{proof}

\begin{rem}
\label{rem: explicitdomainofnormalexpdiffeo}
As shown in the reference given above, the set $\mathcal{D}^+$ is explicitly given by
\begin{align}
    \mathcal{D}^+ = \{(t,\x) \in [0,\infty) \times \Sigma \mid \x \in \Sigma, t \in [0,c_{\Sigma}^+(\x))\}^\circ.
\end{align}
\end{rem}

\begin{lem}
\label{lem: normalexpslices} 
Consider the diffeomorphism $\exp_{\Sigma}^+: \mathcal{D}^+ \to \mathcal{M}^+$. Then for any $t\geq 0$ the sets
$(\{t\} \times \Sigma) \cap \mathcal{D}^+ \subseteq \mathcal{D}^+$ and for any $t\geq 0$, $ \mathcal{S}_t:=\{\tau_{\Sigma} = t\} \cap
\mathcal{M}^+ \subseteq \mathcal{M}^+$ are either empty or embedded hypersurfaces (in the respective supersets).
Moreover, if nonempty,  $\mathcal{S}_t$ is spacelike and the restriction 
of the exponential map $\exp_{\Sigma}^+: (\{t\} \times \Sigma) \cap \mathcal{D}^+ \to \mathcal{S}_t$ is a diffeomorphism.
\end{lem}
\begin{proof}
Since $\tau$ is smooth on $\mathcal{M}^+$ with timelike gradient by \cite[Prop.\ 3.2.33]{treude2011ricci}, $\mathcal{S}_t$ is an embedded, spacelike hypersurface. 

Since $\exp_{\Sigma}^+: \mathcal{D}^+ \to \mathcal{M}^+$ is a diffeomorphism, for any $q\in \mathcal{S}_t \subseteq \mathcal{M}^+$ there exists a point $(t',\x)\in \mathcal{D}^+$ with $\exp_{\Sigma}^+(t',\x)=q$, so by Rem.~\ref{rem: explicitdomainofnormalexpdiffeo}, $t'<c_\Sigma^+(\x)$ so $s\mapsto \exp_{\Sigma}^+(s,\x)$ is maximizing until parameter $t'$ and hence $t'=\tau_{\Sigma}(q)=t$. So $\mathcal{S}_t\subseteq \exp_\Sigma^+((\{t\} \times \Sigma)\cap \mathcal{D}^+)$. On the other hand, any $(t,\x)\in \mathcal{D}^+$ satisfies $t<c_\Sigma^+(\x)$ by Rem.~\ref{rem: explicitdomainofnormalexpdiffeo} and hence $\tau_{\Sigma}(\exp_{\Sigma}^+(t,\x)) =t$. So the restriction of $\exp_\Sigma^+$ to  $(\{t\} \times \Sigma)\cap \mathcal{D}^+$ is bijective. It is a diffeomorphism because $\exp_\Sigma^+:\mathcal{D}^+ \to \mathcal{M}^+$ is one.
\end{proof}

\begin{lem}
\label{lem: normalexpmeasure}
Consider the diffeomorphism $\exp_{\Sigma}^+:\mathcal{D}^+ \to \mathcal{M}^+$. Then the pullback measure on $\mathcal{D}^+$ of the volume measure $vol_g$ on $\mathcal{M}^+$ under this diffeomorphism (i.e. the standard pushforward measure under its inverse) is of the form
\begin{align}
    \mathcal{A} \, dt \otimes d\sigma,
\end{align}
where $\sigma$ is the volume measure on $\Sigma$ with respect to the Riemannian metric induced by $g$. The infinitesimal area $\mathcal{A} = \mathcal{A}(t,\x):\mathcal{D}^+\to (0,\infty)$ is smooth in $\x$ and 
\begin{align}
\int_{\mathrm{pr}_\Sigma((\exp_\Sigma^+)^{-1}(\mathcal{S}_t))} \mathcal{A}(t,\x) d\sigma(\x) = \sigma_t(\mathcal{S}_t),
\end{align}
where $\sigma_t$ is the Riemannian measure on $\mathcal{S}_t$. In particular, $\mathcal{A}(0,\x) = 1$ for all $\x \in \Sigma$.
\end{lem}
\begin{proof}
This is basically the ``measure version'' of the coarea formula, cf.\ \cite[Prop.\ 3]{treude2013volume}. The coarea formula implies that for any $f \in L^1(J^+(\Sigma))$,
\begin{align}
    \int_{J^+(\Sigma)} f \, dvol_g = \int_{\mathcal{M}^+} f \, dvol_g = \int_0^{\infty} \int_{\mathcal{S}_t} f|_{\mathcal{S}_t} d\sigma_t \, dt.
\end{align}
By definition of the pushforward of a measure, we get that
\begin{align}
    \int_{\mathcal{S}_t} f|_{\mathcal{S}_t} \, d\sigma_t &= \int_{\mathcal{S}_t} (f \circ \exp_{\Sigma}^+) \circ (\exp_{\Sigma}^+)^{-1} \, d\sigma_t \nonumber\\
    &= \int_{(\{t\} \times \Sigma) \cap \mathcal{D}^+} (f \circ \exp_{\Sigma}^+) \, d((\exp_{\Sigma}^+)^{-1}_* \sigma_t).
\end{align}
Since the volume form is locally an $n$-fold wedge, the pushforward of $\sigma_t$ under $(\exp_{\Sigma}^+)^{-1}$ is proportional to the Riemannian measure on $(\{t\} \times \Sigma) \cap \mathcal{D}^+$ (which by identification is just $\sigma$), so there is a smooth (in $\x \in \Sigma$) function $\mathcal{A}(t,\x)$ such that $d((\exp^{+})^{-1}_* \sigma_t) = \mathcal{A}(t,.) \,d\sigma$. Reinserting this into the previous calculation, and then that back into the coarea formula, we get
\begin{align}
    \int_{\mathcal{M}^+} f \, dvol_g = \int_0^{\infty} \int_{(\{t\} \times \Sigma) \cap \mathcal{D}^+} (f \circ \exp_{\Sigma}^+) \mathcal{A}\, d\sigma dt = \int_{\mathcal{D}^+} (f \circ \exp_{\Sigma}^+) \mathcal{A} \, dt\otimes d\sigma.
\end{align}
Since this holds for all integrable $f$, the claim about the measures follows.
\end{proof}

Let us recapitulate: For any $f \in L^1(J^+(\Sigma))$, we have
\begin{equation}
\label{eq: integralinMvsintegralinproduct}
    \int_{\mathcal{M}^+} f \, dvol_g = \int_{\mathcal{D}^+} (f \circ \exp_{\Sigma}^+) \, \mathcal{A} \, dt \otimes d\sigma.
\end{equation}
The same relation is true if we integrate over any measurable subset $C \subseteq \mathcal{D}^+$ and $\exp_{\Sigma}^+(C) \subseteq \mathcal{M}^+$.

\begin{rem}
\label{rem: A=sqrt(det g_t)/sqrt( det g_0)}
If we write $g(t,\x) = g(\exp_{\Sigma}^+(t,\x))$, then, upon choosing a chart in $\Sigma$, it follows that $\mathcal{A}(t,\x) = \frac{\sqrt{\det g(t,\x)}}{\sqrt{\det g(0,\x)}}$, in particular $\mathcal{A}$ is smooth in both variables. This can be seen from the splitting of the metric in $\exp_{\Sigma}^+$-coordinates, i.e.\
\begin{align}
    (\exp_{\Sigma}^+)^* g = -dt^2 + g_t,
\end{align}
where $g_t$ is the Riemannian metric induced by $g$ on $\mathcal{S}_t$. Later, we will be interested in proving comparison results for the infinitesimal area $\mathcal{A}(t,\x)$, which then directly imply more or less identical results for $\sqrt{\det g(t,\x)}$.
\end{rem}

We are interested in regions of $\Sigma$ such that the normal geodesics emanating from there are free of cut points up to a certain length. Working with such regions, rather than $\Sigma$, will allow us to use the normal exponential map in our integral calculations.

\begin{defn}[Future evolutions, regular points]
\label{defn: evolutionsregularpoints}
\begin{enumerate}
    \item[]
    \item For $T > 0$ and $B \subseteq \Sigma$, define the \textit{future evolution} of $B$ up to $T$ via
    \begin{align}
    \Omega_T^+(B):=\{\exp_{\x}(t\vec{n}_{\x}): \x \in B, t \in [0,T]\cap [0,c_\Sigma^+(\x))\}.
    \end{align}
    \item For $T,\eta > 0$, the \textit{set of $T+\eta$-regular points} $\Regplus(T) \subseteq \Sigma$ is defined to be the set of all $\x \in \Sigma$ such that $s^+(\x)>T+\eta$ and $\exp_{\Sigma}^+(.,\x):[0,T+\eta] \to M$ is $\Sigma$-maximizing.
\end{enumerate}
\end{defn}

Usually, $\eta$ will be much smaller than $T$ and is there to guarantee a lack of cut points for a short time beyond $T$.

\begin{rem}[On $\Omega_T^+$, $\Regplus(T)$]
\label{rem: onevolutionsregularpoints}
\begin{enumerate}
\item[]
\item By definition, $\Regplus(T) = (c_{\Sigma}^+)^{-1}([T+\eta,\infty])$
so $\Regplus(T)$ is measurable by the lower semicontinuity of $c_{\Sigma}^+$.
\item It is easy to see that for any $B \subseteq \Regplus(T)$,
\begin{align}
    \mathcal{M}^+ \supseteq \Omega_T^+(B)  = \exp_{\Sigma}^+ ([0, T] \times B \}).
\end{align} 
\end{enumerate}
\end{rem}

\begin{defn}[Congruence of $\Sigma$]
\label{defn: congruencevf}
Consider the diffeomorphism $\exp_{\Sigma}^+:\mathcal{D}^+ \to \mathcal{M}^+$. Let $Y \in \mathfrak{X}(\R \times \Sigma)$ denote the canonical vector field $Y:(t,\x) \mapsto (t,0) \in \R \times T_\x\Sigma$. Now consider the restriction of $Y$ as a vector field on $\mathfrak{\mathcal{D}^+}$ and let $U \in \mathfrak{X}(\mathcal{M}^+)$ denote its pushforward under the diffeomorphism $\exp_{\Sigma}^+$. We call $U$ the \textit{future unit congruence (vector field)} of $\Sigma$. Explicitly, if $p = \exp_{\Sigma}^+(t,\x) \in \mathcal{M}^+$, then
\begin{align}
    U(p) = \frac{d}{ds}\Big|_{s=0} \exp_{\x}((t+s)\vec{n}_{\x}),
\end{align}
i.e.\ $U$ assigns to $p \in \mathcal{M}^+$ the tangent vector (at $p$) of the unique normal unit speed geodesic from $\Sigma$ that passes through $p$.
\end{defn}

\section{Forward and backward comparison results}
\label{sec: forwardbackwardcomparison}

In the following, we discuss comparison results for the infinitesimal area $\mathcal{A}(t,\x)$ as well as the d'Alembert operator of the (signed) time separation.
As they will be used in various estimates, we give a detailed treatment of all the comparison constants.

In addition to the assumptions from the previous section, we assume that $(M,g)$ satisfies $\Ric_g(X,X) \geq n\kappa$ for all unit timelike vectors $X$ and some 
$\kappa<0$. Further the mean curvature $H$ of $\Sigma$ with respect to the future normal $\vec{n}$ satisfies $H \leq \beta < 0$.
\medskip

We start by proving an infinitesimal version of an area comparison result given in  \cite[Thm.\ 8]{treude2013volume}, which allows for a direct estimation of
the time evolution of the area element by the modified distance functions of constant curvature spaces.

\begin{rem}[Comparison geometries] \label{rem:compgeoms}
For completeness we briefly summarize the model warped product geometries $(a,b)\times_f (\Sigma,h)$ with $(\Sigma,h)$ a complete Riemannian manifold and Lorentzian metric $g= -dt^2+f^2 h $ from \cite[Sec.\ 4.2]{treude2013volume}. We want the model geometries to have constant Ricci curvature $\Ric_{g} = -n \kappa g$ . We further require the mean curvature of the $\{t=0\}$-slice in the model geometry to equal $\beta$, hence we have a two-parameter family of model geometries and will henceforth index all quantities in the model geometries with $\kappa$ and $\beta$. As was argued in \cite{treude2013volume}, these requirements force the Riemannian factor to be Einstein and (up to rescaling) uniquely determine $f_{\kappa,\beta}$. 
Tables containing a complete list of all possible comparison warping functions $f_{\kappa,\beta}$ can be found e.g. in \cite[Tab.\ 1]{treude2013volume} or \cite[Tab.\ 1]{graf2016volume}.
Congruent with our assumptions we will only elaborate on the case $\kappa<0$ and $\beta<0$ here, distinguishing three different cases: 
\begin{enumerate}
    \item $\beta > -(n-1) \sqrt{|\kappa|}$, where we can take  $(\Sigma_{\kappa,\beta},h_{\kappa,\beta})$ to be the $(n-1)$-dimensional unit sphere,  in which case $a_{\kappa,\beta}=-\infty$ and $b_{\kappa,\beta}=\infty $,
\begin{equation}\label{eq:comparisonfbetanear0}
        f_{\kappa,\beta}(t) = \frac{1}{\sqrt{|\kappa|}} \cosh(\sqrt{|\kappa|}t + \text{artanh}\left(\frac{\beta}{(n-1)\sqrt{|\kappa|}}\right))
\end{equation}
In this case the mean curvature of the hypersurface $\{t\}\times \Sigma_{\kappa,\beta}$, $H_{\kappa,\beta}(t)$, is given by
    \begin{equation}\label{eq:comparisonHbetanear0}
        H_{\kappa,\beta}(t)= (n-1) \frac{f_{\kappa,\beta}'(t)}{f_{\kappa,\beta}(t)}=(n-1) \sqrt{|\kappa|}\tanh(\sqrt{|\kappa|}t + \text{artanh}\left(\frac{\beta}{(n-1)\sqrt{|\kappa|}}\right)) 
\end{equation}
    \item $\beta = -(n-1) \sqrt{|\kappa|}$, where we can take  $(\Sigma_{\kappa,\beta},h_{\kappa,\beta})$ to be  flat $\R^{n-1}$, in which case $a_{\kappa,\beta}=-\infty$ and $b_{\kappa,\beta}=\infty$ and
    \begin{equation}
         f_{\kappa,\beta}(t) = \exp(-\sqrt{|\kappa|}t)
    \end{equation}
    \item $\beta < -(n-1) \sqrt{|\kappa|}$, where we can take  $(\Sigma_{\kappa,\beta},h_{\kappa,\beta})$ to be  unit $(n-1)$-dimensional hyperbolic space, in which case $a_{\kappa,\beta}=-\infty$ and $b_{\kappa,\beta}= - \frac{1}{\sqrt{|\kappa|}} \text{arcoth}\left(\frac{\beta}{(n-1)\sqrt{|\kappa|}}\right)$ (i.e., the model geometries are future timelike goedesically incomplete), 
    \begin{equation}\label{eq:comparisonfbetaverynegative}
              f_{\kappa,\beta}(t) = \frac{1}{\sqrt{|\kappa|}} \sinh(\sqrt{|\kappa|}t + \text{arcoth}\left(\frac{\beta}{(n-1)\sqrt{|\kappa|}}\right)).
    \end{equation}
  In this case the mean curvature of the hypersurface $\{t\}\times \Sigma_{\kappa,\beta}$ is given by
   \begin{equation}\label{eq:comparisonHbetaverynegative}
        H_{\kappa,\beta}(t)= (n-1) \frac{f_{\kappa,\beta}'(t)}{f_{\kappa,\beta}(t)}=(n-1) \sqrt{|\kappa|}\coth(\sqrt{|\kappa|}t + \text{arcoth}\left(\frac{\beta}{(n-1)\sqrt{|\kappa|}}\right)) 
\end{equation}
\end{enumerate}
\end{rem}

Our first result shows that given a lower Ricci bound and a bound on the mean curvature of the Cauchy surface $\Sigma$, the area element of the evolving surface can be given an upper bound corresponding essentially to the area element in the respective model geometry. 

\begin{lem}[Forward infinitesimal area comparison]
\label{lem:inf_area_comp_forward}
For $(t,\x) \in \mathcal{D}^+$ the infinitesimal area $\mathcal{A}(t,x)$ satisfies
	\begin{align}
		\mathcal{A}(t,\x) \leq \left( \frac{f_{\kappa,H(\x)}(t)}{f_{\kappa,H(\x)}(0)}\right)^{n-1} \cdot \mathcal{A}(0,x)\leq \left( \frac{f_{\kappa,\beta}(t)}{f_{\kappa,\beta}(0)}\right)^{n-1},
	\end{align}
where $H(\x)$ denotes the forward mean curvature of $\Sigma $ at $\x$ (i.e.\ the mean curvature with respect to the future normal $\vec{n}_{\x}$) and $f_{\kappa,\beta}$ is the comparison warping function of Remark \ref{rem:compgeoms}.
\end{lem}
	\begin{proof}
The second inequality is obvious since $H\leq \beta$. In the following, denote by $\mathcal{S}_{A,t}$ the set $\mathcal{S}_t \cap \exp_{\Sigma}^+((\{t\} \times A)\cap \mathcal{D}^+)$
for $A \subseteq \Sigma$. 
By \cite[Thm.\ 8, Eq.\ (15)]{treude2013volume} for any measurable $A \subseteq \Sigma$,
\begin{align}
    \frac{\sigma_t(\mathcal{S}_{A,t})}{\sigma(A)} \leq \left(\frac{f_{\kappa,\tilde{\beta}}(t)}{f_{\kappa,\tilde{\beta}}(0)}\right)^{n-1},
\end{align}
where $\tilde{\beta}$ is any upper bound of $H$ on $A$ and $\sigma_t$ is the Riemannian measure on $\mathcal{S}_t$. The proof of Lemma \ref{lem: normalexpmeasure} shows that if $\{t\} \times A \subseteq \mathcal{D}^+$, then
\begin{align}
    \sigma_t(\mathcal{S}_{A,t}) = \int_A \mathcal{A}(t,\x) \, d\sigma(\x).
\end{align}
Now suppose $(t,\x) \in \mathcal{D}^+$. Choose a sequence of open, relatively compact $\Sigma$-balls $U_i$ around $\x$ that collapse to $\{ \x \}$, then eventually $\{t\} \times U_i \subseteq \mathcal{D}^+$ and by continuity
\begin{align}
    \sigma_t(\mathcal{S}_{U_i,t}) \overset{i \to \infty}{\longrightarrow} \mathcal{A}(t,\x),
\end{align}
which, together with the comparison result referenced above (taking $\tilde{\beta}_i = \max_{\x \in \overline{U}_i} H(\x)$), proves the claim.
	\end{proof}

\begin{rem}
\label{rem: compfunctionmonot}
\noindent
\begin{enumerate} 
        \item The proof of Lemma \ref{lem:inf_area_comp_forward} is based on a result in \cite[Thm. 8]{treude2013volume}, which
        relies on $\Sigma$ being a Cauchy surface (or at least future causally complete) in order to conclude that any point in the evolution
        $\mathcal{S}_{A,t}$ is reached by a maximizing timelike geodesic from $A$. However this is the only place where $\Sigma$ being Cauchy is necessary. If we restrict ourselves to $(t,\x)$ such that there exists an open neighbourhood $U$ of $\x$ in $\Sigma$ and an $\epsilon >0$ such that $\exp^+_{\Sigma}|_{(-\epsilon, t+\epsilon)\times U}$ is a diffeomorphism onto an open subset of $M$, then the argument goes through without assuming that $\Sigma$ is Cauchy. This will be relevant when applying the previous Lemma in proving our backwards comparison result Lemma \ref{lem: infinites_backwards_comp}, as it will be applied to a level set in the evolution of $\Sigma$, which does not have to be a  Cauchy surface in general. 
        \item Note that, if $\kappa <0$ and $\beta \geq -(n-1)\sqrt{|\kappa|}$, then $f_{\kappa,\beta}$ first monotonically decreases to some minimum before monotonically increasing (cf.\ Remark \ref{rem:compgeoms}), i.e., we get for $t \leq T$
	        \begin{align} 
	         \mathcal{A}(t,\x) \leq \max\left\{ 1, \left( \frac{f_{\kappa,\beta}(T)}{f_{\kappa,\beta}(0)}\right)^{n-1}\right\}=:C^{A_+}(n,\kappa,\beta,T).
	         \end{align}
        \item If $\kappa < 0$ and $\beta < -(n-1)\sqrt{|\kappa|}$, then the model geometry is future timelike geodesically incomplete (cf.\ Remark \ref{rem:compgeoms}) and a standard generalization of Hawking's singularity theorem shows that $(M,g)$ is future timelike geodesically incomplete (see e.g.\ \cite[Thm.\ 4.2]{graf2016volume} or \cite{Galloway:2013cea}).
        Hence, we will mostly be interested in the case $\kappa,\beta < 0$ and $\beta \geq -(n-1)\sqrt{|\kappa|}$.
    \end{enumerate}
\end{rem}

In order to apply the area comparison results and prove a segment type inequality,
we need to be able to estimate the area element of $\mathcal{S}_t$ from below. This is done in the next Lemma and uses what we call \textit{backwards comparison}. This simply means that instead of estimating 
the area element from above in terms of $\kappa$ and
the mean curvature of $\Sigma$, we use the argument backwards (time reversed) from the evolution of $\Sigma$ at time $T$. 
As we still have control over the mean curvature at that time, we are able to estimate the area. Importantly, this procedure will result in a (non-trivial) bound on the area element from {\em below} in terms of $\kappa$, $\eta$ and $T$.

For the following let $H:\mathcal{M}^+ \to \mathbb{R}$ denote the map assigning to each $p\in \mathcal{M}^+$ the future mean curvature of the smooth spacelike hypersurface $\mathcal{S}_{\tau_{\Sigma}(p)}$ at the point $p$. Clearly, $H|_{\Sigma}$ is then the mean curvature of $\Sigma$, hence there is no confusion in denoting both by the same symbol $H$.

\begin{lem}[Backward infinitesimal area comparison]
\label{lem: infinites_backwards_comp}
For all $T,\eta > 0$ there is a constant $C^{A-}=C^{A-}(n,\kappa,T,\eta)$ such that for all $t \in [0,T]$  and all $\x \in \Regplus(T)$,
\begin{align}
    \mathcal{A}(t,\x) \geq \left( \frac{f_{\kappa,-H(\exp_{\x}(t \vec{n}_{\x}))}(0)}{f_{\kappa,-H(\exp_{\x}(t \vec{n}_{\x}))}(t)}\right)^{n-1} \geq \left( \frac{f_{\kappa,\tilde{\beta}(n,\kappa, \eta)}(0)}{f_{\kappa,\tilde{\beta}(n,\kappa, \eta)}(T)}\right)^{n-1}=:C^{A-}(n,\kappa, T, \eta),
\end{align}
where $\tilde{\beta}(n, \kappa, \eta):= (n-1) \sqrt{|\kappa|} \coth(\eta \sqrt{|\kappa|}) $.
\end{lem}
\begin{proof}
The first inequality follows from changing the time orientation and running the argument of Lemma \ref{lem:inf_area_comp_forward} backwards starting at $\exp_{\x}(t \, \vec{n}_{\x})$.

It remains to argue that $-H(\exp_{\x}(t \vec{n}_{\x}))\leq (n-1) \sqrt{|\kappa|} \coth(\eta \sqrt{|\kappa|})$ for any
$\x \in \Regplus(T)$. Assume to the contrary that
\begin{align}
-H(\exp_{\x}(t \vec{n}_{\x}))> (n-1) \sqrt{|\kappa|} \coth(\eta \sqrt{|\kappa|}),
\end{align}
then by continuity
of $\coth$ there is $\varepsilon > 0$ with $\varepsilon < \eta$ such that this inequality holds with $\eta$ replaced by $\eta - \varepsilon$, i.e., $H(\exp_{\x}(t \vec{n}_{\x}))<- (n-1)
\sqrt{|\kappa|} \coth((\eta - \varepsilon) \sqrt{|\kappa|})$, for some $\x \in \Regplus(T)$. Then, the future timelike unit speed geodesic
starting at $\exp_{\x}(t \vec{n}_{\x})$ and orthogonally to $\mathcal{S}_t$ is given by $s\mapsto \exp_{\x}((t+s) \vec{n}_{\x}) $. Note that this geodesic maximizes
the distance to $\mathcal{S}_t$ for all $s\in [0,\eta]$. By standard (forward) mean
curvature/d'Alembertian comparison starting from $\mathcal{S}_t$ (cf. \cite[Thm.\ 7]{treude2013volume}), $H(\exp_{\x}((t+s) \vec{n}_{\x})) \leq
H_{\kappa, H(\exp_{\x}(t \vec{n}_{\x}))}(s)$\footnote{As in Remark \ref{rem:compgeoms} here $H_{\kappa,\beta}(t)$ is the mean curvature of the spacelike hypersurface $\{t\} \times \Sigma_{\kappa,\beta}$ in the comparison geometry $(a_{\kappa,\beta},b_{\kappa,\beta}) \times_{f_{\kappa,\beta}} (\Sigma_{\kappa,\beta},h_{\kappa,\beta})$.} for all $s\in [0,\eta]$ and since $H(\exp_{\x}(t \vec{n}_{\x}))<- (n-1) \sqrt{|\kappa|} \coth((\eta - \varepsilon) \sqrt{|\kappa|})=:\beta_0$, actually even 
\begin{align}
H(\exp_{\x}((t+s) \vec{n}_{\x})) \leq
H_{\kappa, \beta_0}(s).
\end{align}
Note that $\beta_0<-(n-1) \sqrt{|\kappa|}$ , so checking with Remark \ref{rem:compgeoms}, specifically equation \eqref{eq:comparisonHbetaverynegative},  
\begin{align}
H_{\kappa,\beta_0 }=(n-1)\sqrt{|\kappa|}\coth(\sqrt{|\kappa|}s+ \coth^{-1}\left(\frac{\beta_0}{(n-1)\sqrt{|\kappa|}}\right)),
\end{align}
so the previous inequality becomes
\begin{align}
H(\exp_{\x}((t+s) \vec{n}_{\x})) \leq (n-1)\sqrt{|\kappa|}\coth(\sqrt{|\kappa|}s-(\eta - \varepsilon)\sqrt{|\kappa|})  .
\end{align}
Now we see that the right hand side diverges to $-\infty $ as $s\uparrow (\eta - \varepsilon) < \eta$, but the left hand site must remain finite for all $s\in [0,\eta)$ since $t\mapsto \exp_{\x}(t \vec{n}_{\x})$ doesn't have any $\Sigma$-cut points  before $T+\eta$. 
\end{proof}

Closely connected to the evolution of the area element is the d'Alembertian of the time separation function $\tau_{\Sigma}$. Heuristically speaking it measures the expansion of the volume spanned by neighbouring geodesics of the congruence and usually comparison results for the d'Alembertian are established first and then used to show area comparison results. However, as we are basing our arguments upon already established comparison results we can treat both cases independently.
As was the case for the area, it allows for an estimation from above with 
respect to the constant curvature comparison space which we include and which was shown in \cite[Thm.\ 7]{treude2013volume}.
The lower bound again follows from a backwards
comparison argument.

\begin{lem}[Forward d'Alembert comparison]
\label{lem: forwarddalembertcomparison}
For $(t,\x) \in \mathcal{D}^+$ we have that
\begin{align}
    \Box_g \tau_{\Sigma}(\exp_{\Sigma}^+(t,\x))= H(\exp_{\Sigma}^+(t,x)) \leq H_{\kappa,\beta}(t).
\end{align}
\end{lem}

\begin{lem}[Backward d'Alembert comparison]
\label{lem: backwarddalembertcomparison}
For any $T,\eta > 0$, $t \in [0,T]$ and $\x \in \Regplus(T)$, we have that
\begin{align}
    \Box_g \tau_{\mathcal{S}_T}(\exp_{\Sigma}^+(t,\x)) = -H(\exp_{\Sigma}^+(t,x)) \leq (n-1)\sqrt{|\kappa|}\coth(\eta \sqrt{|\kappa|}),
\end{align}
where $\tau_{\mathcal{S}_T}$ is the time separation from the spacelike hypersurface $\mathcal{S}_T$ measured backwards.
\begin{proof}
This was shown in the proof of Lemma \ref{lem: infinites_backwards_comp}.
\end{proof}
\end{lem}

\begin{rem}[On the comparison constants]
\label{rem: comparisonconstants}
As already stated in the third point of Remark \ref{rem: compfunctionmonot}, we are mainly interested in the case $0 > \beta \geq -(n-1)\sqrt{|\kappa|}$, so
\begin{align}
    0 < \frac{|\beta|}{(n-1)\sqrt{|\kappa|}} \leq 1.
\end{align}
For convenience we first exclude the case $\beta = -(n-1)\sqrt{|\kappa|}$ and treat it separately later. While this case won't be particularly relevant in our applications as we may replace any bound $\beta<0$ on $H$ from above with some slightly bigger bound $\tilde{\beta} < 0$, including it does not pose any difficulties. Checking with Remark \ref{rem:compgeoms} then yields
\begin{align}
    f_{\kappa,\beta}(t) = \frac{1}{\sqrt{|\kappa|}} \cosh(\sqrt{|\kappa|}t + \text{artanh}\left(\frac{\beta}{(n-1)\sqrt{|\kappa|}}\right)).
\end{align}
Hence, after some manipulations, the forward infinitesimal area comparison constant $C^{A+}=C^{A+}(n,\kappa,\beta,T)$ can be written as
\begin{align}
    C^{A+}(n,\kappa,\beta,T) = \left(\frac{f_{\kappa,\beta}(T)}{f_{\kappa,\beta}(0)}\right)^{n-1} = \left( \frac{\beta}{(n-1)\sqrt{|\kappa|}} \sinh(\sqrt{|\kappa|}T) + \cosh(\sqrt{|\kappa|}T)\right)^{n-1}.
\end{align}
Next, we want to derive an explicit form of the backward infinitesimal area comparison constant $C^{A-} = C^{A-}(n,\kappa,T,\eta)$, cf.\ Lemma \ref{lem: infinites_backwards_comp}. First note that
\begin{align}
    \frac{\tilde{\beta}(n,\kappa,\eta)}{(n-1)\sqrt{|\kappa|}} = \coth\left(\eta \sqrt{|\kappa|}\right) > 1
\end{align}
because $\coth(x) > 1$ for any $x > 0$. Again checking  with Remark \ref{rem:compgeoms}, the relevant function in this case is
\begin{align}
    f_{\kappa,\tilde{\beta}}(t) = \frac{1}{\sqrt{|\kappa|}} \sinh\left(\sqrt{|\kappa|}t + \mathrm{arcoth} \left(\frac{\tilde{\beta}}{(n-1)\sqrt{|\kappa|}}\right)\right).
\end{align}
Inserting $\tilde{\beta}=(n-1)\sqrt{|\kappa|} \coth\left(\eta \sqrt{|\kappa|}\right) $, we get
\begin{align}
    C^{A-}(n,\kappa,T,\eta) = \left(\frac{f_{\kappa,\tilde{\beta}}(0)}{f_{\kappa,\tilde{\beta}}(T)}\right)^{n-1} = \left(\frac{\sinh(\eta \sqrt{|\kappa|})}{\sinh(\sqrt{|\kappa|}(T + \eta))}\right)^{n-1}.
\end{align}
Next we consider the forward d'Alembert comparison constant $C^{\square+} = C^{\square+}(n,\kappa,\beta,T)$, by which we mean a uniform upper bound on the quantity $H_{\kappa,\beta}(t)$ appearing in Lemma \ref{lem: forwarddalembertcomparison}. By monotonicity of the comparison function $H_{\kappa,\beta}$, we get
\begin{align}
   H_{\kappa, \beta}(t) \leq  H_{\kappa, \beta}(T) &= (n-1) \sqrt{|\kappa|}\tanh(\sqrt{|\kappa|}T + \text{artanh}\left(\frac{\beta}{(n-1)\sqrt{|\kappa|}}\right)) \\
   &= \frac{(n-1) \sqrt{|\kappa|}\tanh(\sqrt{|\kappa |}T)+\beta }{1+\tanh(\sqrt{|\kappa|}T) \frac{\beta}{(n-1)\sqrt{|\kappa|}} } =:C^{\square+}.
\end{align}
 The backward d'Alembert comparison constant $C^{\square-} = C^{\square-}(n,\kappa,\eta)$ was already given in Lemma \ref{lem: backwarddalembertcomparison}, namely
\begin{align}
    C^{\square-} = (n-1)\sqrt{|\kappa|}\coth(\eta \sqrt{|\kappa|}).
\end{align}
Note that, contrary to $C^{A+},C^{A-}$ and $C^{\square-}$ which are always strictly positive, $C^{\square+}$ can be negative, zero or positive depending on $n,\kappa, \beta,T$.

We now briefly consider $\beta= -(n-1)\sqrt{|\kappa|}$: Clearly the expressions derived for $C^{A+}(n,\kappa,\beta,T)$ and $ C^{\square+}(n,\kappa,\beta,T)$ above continuously extend to $\beta = -(n-1)\sqrt{|\kappa|}$, where they become $\exp(-\sqrt{|\kappa|}T)^{n-1}$ and $-(n-1)\sqrt{|\kappa|}$, respectively, which coincides with the bounds one gets when retracing the above steps for the model geometry in the $\beta= -(n-1)\sqrt{|\kappa|}$ case, where $f_{\kappa,\beta}(t)=\exp(-\sqrt{|\kappa|}t)$ (cf.\ Remark \ref{rem:compgeoms}).
\end{rem}

\begin{rem}[On the d'Alembert operator]
\label{rem: dAlembertop}
Consider the diffeomorphism $\exp_{\Sigma}^+:\mathcal{D}^+ \to \mathcal{M}^+$. We have already noted that the metric $g$ splits in $\exp_{\Sigma}^+$-coordinates, i.e.\
\begin{align}
\label{eq: splitmetric}
    (\exp_{\Sigma}^+)^*g = -dt^2 + g_t,
\end{align}
where $g_t$ is the restriction of $g$ to the spacelike hypersurface $\mathcal{S}_t \subseteq \mathcal{M}^+$.
We choose coordinates $x = (x^j)$ on $\Sigma$, and thus also coordinates $(t,x^j) = (y^{\mu})$ on $\mathcal{M}^+$.
Now suppose $F$ is a smooth function on $\mathcal{M}^+$. We write $\dot{F}=\partial_t F$. Then
\begin{align}
    \Box_g (F^2) &= -\frac{1}{\sqrt{|\det g|}} \frac{\partial}{\partial y^{\mu}}\left(g^{\mu \nu} \sqrt{|\det g} \frac{\partial F^2}{\partial y^{\nu}}\right) \nonumber\\
    &= -\frac{1}{\sqrt{|\det g|}} \frac{\partial}{\partial t} \left(\underbrace{g^{00}}_{=-1} \underbrace{\sqrt{|\det g|}}_{=\sqrt{\det g_t}} \frac{\partial F^2}{\partial t}\right) - \frac{1}{\sqrt{|\det g|}} \frac{\partial}{\partial x^i} \left(g_t^{ij} \sqrt{|\det g|} \frac{\partial F^2}{\partial x^j}\right) \nonumber \\
    &= 2\dot{F}^2 + 2F\ddot{F} + \frac{\partial_t \sqrt{\det g_t}}{\sqrt{\det g_t}} 2F\dot{F} - \Delta_{g_t}(F^2).
\end{align}
For the time separation, we have that
\begin{align}
    \tau_{\Sigma}(\exp_{\Sigma}^+(t,\x)) = t.
\end{align}
One then immediately checks (using the previous calculation, (\ref{eq: splitmetric}) and Remark \ref{rem: A=sqrt(det g_t)/sqrt( det g_0)}) that
\begin{align}
    \Box_g (\tau_{\Sigma} \circ \exp_{\Sigma}^+) = \frac{\partial_t \sqrt{\det g_t}}{\sqrt{\det g_t}}  = \frac{\partial_t \mathcal{A}}{\mathcal{A}}.
\end{align}
Similarly, the backward time separation $\tau_{\mathcal{S}_T}$ on the $T$-slice satisfies (for $t \leq T$)
\begin{align}
    \tau_{\mathcal{S}_T}(\exp_{\Sigma}(t,\x)) = T - t.
\end{align}
As before,
\begin{align}
    \Box_g(\tau_{\mathcal{S}_T} \circ \exp_{\Sigma}) = -\frac{\partial_t \sqrt{\det g_t}}{\sqrt{\det g_t}}  = - \frac{\partial_t \mathcal{A}}{\mathcal{A}}.
\end{align}
Hence, d'Alembert comparison yields
\begin{align}
     \left| \frac{\partial_t \mathcal{A}}{\mathcal{A}} \right| = \left| \frac{\partial_t \sqrt{\det g_t}}{\sqrt{\det g_t}} \right| \leq \max(C^{\square-},|C^{\square+}|) = C^{\square-}=:C^{\square}_{\max}
\end{align}
on $\Omega_T^+(\Regplus(T))$.
\end{rem}

Using the comparison results from the previous section one can derive estimates for integrals along geodesics from estimates on volume integrals.
The approach taken is similar in spirit to the segment inequality used in the almost splitting theorem by Cheeger--Colding (\cite{cheeger1996lower}; see also \cite{sprouse2000integral}).
This technique is then used to obtain Ricci bounds along single geodesics from world volume inequalities, i.e.\ certain curvature integrals over a spacetime volume. To that end we define a map $\mathcal{F}^T_f$ assigning to a point $\x$ on $\Sigma$ the integral of a non-negative $f$ over the future directed normal geodesic to $\Sigma$ starting at $\x$ up to a time $T$.

\begin{defn}
\label{defn: Ff}
Let $f \geq 0$ be any nonnegative, continuous function on $M$ and $T > 0$. Define
\begin{align}
    \mathcal{F}_f^T: \Sigma \to [0,\infty],\quad \mathcal{F}^T_f(\x):=\int_0^{\min(T,s^+(\x))} f(\exp_{\Sigma}^+(t,\x)) dt.
\end{align}
\end{defn}

Next, we give a segment type inequality, i.e.\ a result inferring bounds on certain 
finite line integrals of $f$ from a region $B\subseteq \Regplus$ from bounds on the integral of $f$ over the ``tube'' $\Omega_T^+(B)$.

\begin{prop}[Segment type inequality]
\label{prop: segmentineq}
Let $f$ be any continuous, nonnegative function on $M$ and $T,\eta > 0$. Then for any measurable subset $B \subseteq \Regplus(T)$ such that $0 < \sigma(B) < \infty$ we have that
\begin{align}
    \inf_{\x \in B} \mathcal{F}^T_f \leq \frac{1}{C^{A-} \sigma(B)} \int_{\Omega_T^+(B)} f \, dvol_g,
\end{align}
where $C^{A-}= C^{A-}(n,\kappa,T,\eta)$ is the backward area comparison constant and $\sigma$ is the (Riemannian) volume measure on $\Sigma$.

\end{prop}
\begin{proof}
Using everything discussed so far in a straightforward calculation:
\begin{align}
    \int_{\Omega_T^+(B)} f \, dvol_g &= \int_{B} \int_0^T f(\exp_{\Sigma}^+(t,\x)) \mathcal{A}(t,\x) dt d\sigma(\x) \nonumber \\
    &\geq C^{A-} \int_{B}\int_0^T f(\exp_{\Sigma}^+(t,\x)) dt d\sigma(\x) \nonumber \\
    &\geq C^{A-} \sigma(B) \inf_{\x \in B} \mathcal{F}^T_f.
\end{align}
\end{proof}

\section{The singularity theorems}
\label{sec: singularitytheorems}

We are now ready to present a typical singularity theorem which can be proven using the techniques we have developed.
We decompose a function $f$ into a sum of its positive and negative parts: $f = f_+ + f_-$ where $f_+ \geq 0$ and $f_- \leq 0$  and write $\Ric_-(U,U)$ for the negative part of $p\mapsto \Ric(U_p,U_p)$.

We shall first state a singularity-theorem adjacent result which has the advantage of following very directly from Proposition \ref{prop: segmentineq}.

\begin{theorem}
\label{thm: singthmbasedonsegmineq}
Let $(M,g)$ be a globally hyperbolic spacetime with a smooth, spacelike Cauchy surface $\Sigma \subseteq M$ such that $\Ric(v,v) \geq n\kappa$ for all unit timelike $v \in TM$ and $H \leq \beta$, where $\kappa,\beta < 0$, $\beta \geq -(n-1)\sqrt{|\kappa|}$. 
Let $B\subseteq \Sigma$ with $ 0<\sigma(B)< \infty$. If 
    \begin{align} \label{eq:Ricminus1}
        \frac{1}{\sigma(B)} \int_{\Omega_T^+(B)} |\Ric_-(U_p,U_p)| \, dvol_g(p) \leq C^{A-}(n,\kappa,\eta,T) K,
    \end{align}
for any $T>0,\eta>0, K>0$ satisfying $ K<|\beta|-\frac{n-1}{T} $, then $B \not\subseteq \Regplus(T)$.
\end{theorem}

\begin{proof}
Assume to the contrary that $B \subseteq \Regplus(T)$. Set $f(p):=|\Ric_-(U_p,U_p)|$, then Proposition \ref{prop: segmentineq} and the assumed estimate \eqref{eq:Ricminus1} imply that
\begin{align}
    \inf_{x \in B} \mathcal{F}_f^T \leq K.
\end{align}
Let $\varepsilon > 0$ be so small that $K + \varepsilon < |\beta|$, then there exists $\x \in B$ such that
\begin{align}
    \mathcal{F}_f^T(\x)=\int_0^T f(\exp_{\x}(t\vec{n}_{\x})) dt \leq K + \varepsilon.
\end{align}

Now the result follows directly from classical variational methods: Denoting by $c(t):=\exp_{\x}(t \vec{n}_{\x})$ the normal geodesic starting in $\x$, by assumption $c$ maximizes the $\Sigma$-time separation up to $p:=\exp_{\x}(T\vec{n}_{\x}) \in I^+(\Sigma)$. Let $E_i(t)$ be spacelike parallel orthonormal vector fields along $c$ such that $\{\dot{c},E_1,\dots,E_{n-1}\}$ is an orthonormal frame along $c$. Consider variations of $c$ starting in $\Sigma$ and ending in $p$ defined by variational vector fields $V_i(t):=(1-t/T)E_i(t)$. Denote by $L_i(t)$ the length functional of the corresponding variation.
Using the second variation formula it is a straightforward calculation (done e.g.\ in section 2 of \cite{Fewster:2019bjg} with a different sign convention) to arrive at
\begin{align}
   \sum_{i=1}^{n-1} L_i''(0) &= \left\langle \dot{c}(0), \sum_{i=1}^{n-1} II(E_i(0),E_i(0)) \right\rangle + \int_0^T \left(1- \frac{t}{T}\right)^2 \Ric(\dot{c},\dot{c}) dt - \frac{n-1}{T}.
\end{align}
Note that the first term on the right hand side is equal to $-H(\x)$ (recall that we follow the conventions of \cite{treude2013volume}). Since $c$ maximizes the $\Sigma$-time separation up to $p$, we get
\begin{align}
    0 \geq \sum_{i=1}^{n-1} L_i''(0) &= -H(\x) + \int_0^T \left(1- \frac{t}{T}\right)^2 \Ric(\dot{c},\dot{c}) dt - \frac{n-1}{T} \nonumber \\
    &\geq |\beta| - \int_0^T \left(1 - \frac{t}{T}\right)^2 f(c(t)) dt - \frac{n-1}{T} \nonumber \\
    &\geq |\beta| - \int_0^T f(c(t))dt - \frac{n-1}{T} \nonumber \\
    &\geq |\beta| - K - \varepsilon - \frac{n-1}{T},
\end{align}
which implies that
\begin{align}
    T \leq \frac{n-1}{|\beta| - K - \varepsilon}.
\end{align}
Taking $\varepsilon \to 0$ we see that this contradicts our assumption on the relationship between $T$, $K$ and $|\beta|$.
\end{proof}

While this is not a singularity theorem per se, as the conclusion of $B\not \subseteq \Regplus$ only tells us that either a geodesic stops existing {\em or} maximizing the distance at or before $T+\eta$, its structure is very similar:
Given the usual assumptions of singularity theorems (on $\Ric$, causality and the mean curvature of a Cauchy surface) it implies bounds on regions evolving from $\Sigma$ on which $\exp$ stays injective (even a diffeomorphism).\\

Let us now state a singularity theorem which still employs the segment inequality type methods, but assumes energy conditions inspired by QEIs.
It will be used to prove geodesic incompleteness for the case of the non-minimally coupled classical scalar field in Sec.~\ref{sec: strongenergyinequality}.

\begin{theorem}
\label{theorem: m=1}
    Let $(M,g)$ be a smooth globally hyperbolic spacetime of dimension $n > 2$, $\Sigma \subseteq M$ a smooth, spacelike Cauchy surface. Suppose the following hold:
    \begin{enumerate}
        \item[(i)] There is $\kappa < 0$ such that $\Ric(v,v) \geq n\kappa$ for all unit timelike $v \in TM$.
        \item[(ii)] There exist $Q_1, Q_2 > 0$ such that 
        for all $F \in C^{\infty}_c(\mathcal{M}^+)$,  
        we have the following integral bound on the Ricci tensor:
        \begin{align}\label{eqn:condi}
          \int_{M}
          \Ric(U_p,U_p) F(p)^2 \, dvol_g(p) \geq -Q_1\|F\|_{L^2(M)}^2 - Q_2 \|\dot{F}\|_{L^2(M)}^2.
        \end{align}
        \item[(iii)] the mean curvature of $\Sigma $ satisfies
			\be
			-H  \geq \min\left\{(n-1)\sqrt{|\kappa|}\coth(\sqrt{|\kappa|}\tau ) ,\nu_*(n,\kappa, \tau) \right\}
			\ee			
			everywhere on $\Sigma$, where 
		\bea
		\label{eqn:nustar}
		  &&\nu_*(n,\kappa,\tau):=\min_{\{(T,\eta): T+\eta=\tau\}}\min_{ \tau_0\in (0,T) } \bigg\{ \left(Q_1 + Q_2\frac{(C^{\square}_{\max}(n,\kappa, \eta))^2}{4} + Q_2\frac{C^{\square}_{\max}(n,\kappa, \eta)}{2T}\right) \frac{T}{3} \nonumber \\ && \qquad +Q_2 \left(1+ \frac{TC^{\square}_{\max}(n,\kappa, \eta)}{2}\right) \left(\frac{1}{\tau_0}+\frac{1}{T-\tau_0}\right)+ n |\kappa| \tau_0 \frac{2}{3} +\frac{n-1}{T-\tau_0} \bigg\} 
		  \,.\eea
  
    \end{enumerate}
     
    Then no future-directed timelike curve emanating from $\Sigma$ has length greater than $\tau $ and $(M,g)$ is future timelike geodesically incomplete.
    \begin{proof}
  For $\min((n-1)\sqrt{|\kappa|}\coth(\sqrt{|\kappa|}\tau),\nu_*) =(n-1)\sqrt{|\kappa|}\coth(\sqrt{|\kappa|}\tau)$ this follows from classical comparison geometry techniques, cf.~e.g.~\cite[Thm.~4.2]{graf2016volume} (note that $H\leq  (n-1)\sqrt{|\kappa|}\coth(-\sqrt{|\kappa|}\tau)$ implies $b_{\kappa,\beta}=\tau$).
  
  Let us therefore assume that $\min((n-1)\sqrt{|\kappa|}\coth(\sqrt{|\kappa|}\tau),\nu_*) =\nu_*$. Now assume for a contradiction that  there is a future-directed timelike curve emanating from $\Sigma$ of length greater than $\tau=T+\eta$. Then $\Regplus(T)$ is nonempty. 
  Given any $f \in C^{\infty}_c(\R)$ with support contained in $[0,T]$, we pick an arbitrary, non-zero $h \in C^{\infty}_c(\Sigma)$ with support contained in $\Regplus(T)$ and consider the function $F(t,\x):=f(t)h(\x)(\mathcal{A}(t,\x))^{-1/2}$ (in exponential coordinates). Then $\|F\|_{L^2(M)}^2 = \|f\|_{L^2(\R)}^2 \|h\|_{L^2(\Sigma)}^2$ and
  
    \begin{align}
        \dot{F}(t,\x) = \dot{f}(t) h(\x) \mathcal{A}(t,\x)^{-1/2} -\frac{1}{2} f(t)h(\x) \mathcal{A}(t,\x)^{-3/2} \partial_t \mathcal{A}(t,\x),
    \end{align}
    hence
    \begin{align}
        \dot{F}(t,\x)^2 &= \dot{f}(t)^2 h(\x)^2 \mathcal{A}(t,\x)^{-1} + \frac{1}{4} f(t)^2 h(\x)^2 \mathcal{A}(t,\x)^{-3} (\partial_t \mathcal{A}(t,\x))^2 \nonumber \\
        &\phantom{aa}- \dot{f}(t) f(t) h(\x)^2 \mathcal{A}(t,\x)^{-2} \partial_t \mathcal{A}(t,\x).
    \end{align}
By d'Alembert comparison, we can estimate $$|\partial_t \mathcal{A} \cdot \mathcal{A}^{-1}| \leq C^{\square}_{\max}(n,\kappa, \eta):=\max(|C^{\square+}(n,\kappa, \eta)|,C^{\square-}(n,\kappa, \eta)) = C^{\square-}(n,\kappa, \eta),$$ as in Remark \ref{rem: dAlembertop}. Thus, dropping the dependence of $C^{\square}_{\max}$ on $n,\kappa,\eta$ for now, and using the classical inequality
\be
\label{eqn:elineq}
ab \leq \varepsilon a^2 + \frac{1}{4 \varepsilon} b^2 \,,
\ee
for $a,b \geq 0$ and $\varepsilon:=T/2$ we have
    \bea
        \|\dot{F}\|_{L^2(M)}^2 &\leq& \|\dot{f}\|_{L^2(\R)}^2 \|h\|_{L^2(\Sigma)}^2 + \frac{(C^{\square}_{\max})^2}{4} \|f\|_{L^2(\R)}^2 \|h\|_{L^2(\Sigma)}^2 + \frac{C^{\square}_{\max} }{2T}(T^2\|\dot{f}\|_{L^2(\R)}^2 \nonumber \\
        &&+  \|f\|_{L^2(\R)}^2)\|h\|_{L^2(\Sigma)}^2.
    \eea
On the other hand,
    \begin{align}
        \int_{\Omega_T^+(\Regplus(T))} \Ric(U_p,U_p) F(p)^2 \, dvol_g(p) \leq \|h\|_{L^2(\Sigma)}^2 \sup_{\x \in \mathrm{supp}(h)} \int_0^T \Ric(\dot{c}_{\x}(t),\dot{c}_{\x}(t)) f(t)^2 \, dt,
    \end{align}
 where $\dot{c}_{\x}(t)=U_{\exp_{\Sigma}^+(t,\x)}$. Since $h$ has compact support in $\Regplus(T)$, the supremum on the right hand side is actually a maximum. 
 Canceling the $\|h\|_{L^2(\Sigma)}^2$, by \eqref{eqn:condi} we are left with
    \begin{align}
        &\min_{\x \in \mathrm{supp}(h)} \left(- \int_0^T \Ric(\dot{c}_{\x}(t),\dot{c}_{\x}(t)) f(t)^2 \, dt  \right)\nonumber\\
        &\leq Q_1 \|f\|_{L^2(\R)}^2 + Q_2\left(\|\dot{f}\|_{L^2(\R)}^2 +\frac{(C^{\square}_{\max})^2}{4} \|f\|_{L^2(\R)}^2 + \frac{C^{\square}_{\max}}{2T} \left(T^2\|\dot{f}\|_{L^2(\R)}^2 + \|f\|_{L^2(\R)}^2\right)\right) \nonumber \\
        &= \left(Q_1 + Q_2\frac{(C^{\square}_{\max})^2}{4} + Q_2\frac{C^{\square}_{\max}}{2T}\right)\|f\|_{L^2(\R)}^2 + Q_2 \left(1+ \frac{TC^{\square}_{\max}}{2}\right)\|\dot{f}\|_{L^2(\R)}^2. \label{eqn:line}
    \end{align}
     Fix $\x_0\in \mathrm{supp}(h)$ where the minimum is achieved and let $c:[0,T]\to M$ be the distance-to-$\Sigma$-maximizing geodesic $c_{\x_0}$ starting orthogonally to $\Sigma $ at $\x_0$.

 Since $f$ is still free (the only restriction is that its support is contained in $[0,T]$), we fix any $\tau_0 \in (0,T)$ and -- following the definitions in \cite[Lem.~4.1]{Fewster:2019bjg} for $m=1$ -- choose $f(t)=\bar{f}(t)\phi(t)$ with  
  $$ \bar{f}(t):=\begin{cases} 1\quad\quad\quad\quad\quad t\in[0,\tau_0]\\ 1-\frac{t-\tau_0}{T-\tau_0}\quad\quad t\in[\tau_0, T]
\end{cases} $$
and 
$$ \phi(t):=\begin{cases} \frac{t}{\tau_0} \quad\quad t\in[0,\tau_0]\\ 1 \quad\quad t\in[\tau_0, T].
\end{cases} $$
With this choice we closely follow the methods used in \cite{Fewster:2019bjg}, especially Proposition 2.2 and partly Lemma 4.1 there and will refer to that article for details of the index form methods applied.
Note that $f$ is piece-wise linear vanishing at $0$ and $T$, so even though $f$ is not smooth all our above calculations, in particular \eqref{eqn:condi} and \eqref{eqn:line}, still hold (this is easily seen by approximating $f$ with compactly supported smooth functions whose derivatives converge to the derivative of $f$ in $L^2$).
The next step is to estimate 
     $$J[\bar{f}]= \int_0^T \left(-\Ric(\dot c,\dot c) \bar{f}^2+(n-1)(\dot {\bar{f}})^2\right) dt $$ 
from above (note the different sign convention for $\Ric$ in \cite{Fewster:2019bjg}). Using the identity $\bar f^2= (\phi \bar f)^2+(1-\phi^2)$, we obtain:
     \bea \label{eq:nutau0}
     J[\bar{f}] &=& \int_0^T -\Ric(\dot c,\dot c) \left((\phi(t) \bar{f}(t))^2+(1-\phi(t)^2)\right) dt + (n-1) ||\dot{\bar{f}}(t)||^2 \nonumber\\
 &\leq& \left(Q_1 + Q_2\frac{(C^{\square}_{\max})^2}{4} + Q_2\frac{C^{\square}_{\max}}{2T}\right)\|\phi \bar{f}\|_{L^2(\R)}^2 + Q_2 \left(1+ \frac{TC^{\square}_{\max}}{2}\right)\left\|\frac{d}{dt}(\phi\, \bar{f})\right\|_{L^2(\R)}^2 \nonumber \\ &&+  n|\kappa| \int_0^{\tau_0} (1-\phi^2) dt + 
 \frac{n-1}{T-\tau_0} = \nonumber\\
 &=& \left(Q_1 + Q_2\frac{(C^{\square}_{\max})^2}{4} + Q_2\frac{C^{\square}_{\max}}{2T}\right) \frac{T}{3}+ Q_2 \left(1+ \frac{TC^{\square}_{\max}}{2}\right) \left(\frac{1}{\tau_0}+\frac{1}{T-\tau_0}\right)\nonumber \\
 &&+ n|\kappa| \tau_0 \frac{2}{3} +\frac{n-1}{T-\tau_0}:=\nu_{\tau_0,\eta} \,,
\eea
where we estimated $\Ric(\dot c, \dot c)\geq n \kappa$ and integrated the remaining terms involving $\phi$.

By \cite[Lemma 4.1]{Fewster:2019bjg} there is a focal point along $c$ if $J[\bar{f}]\leq -H|_{\x_0}$.
Note that by \eqref{eq:nutau0} $J[\bar{f}]\leq \nu_{\tau_0,\eta}$. Since this works for any $\tau_0\in (0,T)$ and any $\eta, T>0$ such that $T+\eta=\tau$, we may choose $\tilde\tau_0$ and $\tilde \eta+T$ such that $\nu_{\tilde\tau_0,\tilde\eta}$ becomes minimal (note that such a minimizing $\tau_0$ exists in the open interval because the expression goes to infinity as $\tau_0\to 0,T$).
This leads to $J[\bar{f}] \leq \nu_{\tilde\tau_0,\tilde\eta}= \nu_* \leq -H|_{x_0}$ by the assumption \eqref{eqn:nustar} and there must be a focal point along $c$ which in turn cannot be maximizing up to time $\tau=T+\eta$,
contradicting $\x_0\in \Regplus(T)$.
    \end{proof}
\end{theorem}

\begin{rem}
\label{rem: m>1case}\, 
\begin{enumerate}
    \item 
    The proof of Theorem \ref{theorem: m=1} really only needs the integral bound \eqref{eqn:condi} for all $F\in C_c^\infty(\Omega_T^+(\Regplus(T)))$. However, as that set is in general hard to control we chose to formulate Theorem \ref{theorem: m=1} as is.
    \item In the formulation of Theorem \ref{theorem: m=1} higher derivatives of $F$ could be accommodated for at the price of imposing bounds on the positive and negative parts individually (writing a function $b$ as $b=b^+ + b^-$, $b^+ \geq 0$, $b^- \leq 0$), i.e.\  with bounds of the form ($Q_+,Q_-,Q_+^{(m)},Q_-^{(m)} > 0$)
\begin{align}
    &\int_{\Omega_T^+(\Regplus(T))} \Ric_+(U_p,U_p) F(p)^2 \, dvol_g(p) \geq Q_+ \|F\|_{L^2(M)}^2 + Q_+^{(m)} \|F^{(m)}\|_{L^2(M)}^2,\\
    &\int_{\Omega_T^+(\Regplus(T))} |\Ric_-(U_p,U_p)| F(p)^2 \, dvol_g(p) \leq Q_- \|F\|_{L^2(M)}^2 + Q_-^{(m)} \|F^{(m)}\|_{L^2(M)}^2.
\end{align}
Then one can proceed as above to prove the existence of a singularity: Indeed, choosing $F$ to be of the form $F(t,\x) = f(t) h(\x)$ as before -- just without the extra factor of $\mathcal{A}(t,\x)^{-\frac{1}{2}}$ --, one estimates  $\int \Ric_+ F^2 \, dvol_g$ forwards (using forward comparison) and backwards (using the assumption as well as backward comparison), and similarly for $\int \Ric_- F^2 \, dvol_g$, to get an estimate on $\inf (-\int \Ric f^2)$ as before. From these estimates, it is then clear by following \cite{Fewster:2019bjg} how $\nu_*$ has to be chosen, and the rest of the proof is completely analogous to Theorem \ref{theorem: m=1}.
\end{enumerate}
\end{rem}

\section{The strong energy inequality}
\label{sec: strongenergyinequality}

In this section we want to show that the inequality \eqref{eqn:condi} is satisfied by the non-minimally coupled classical scalar field, an example that violates the SEC. We additionally discuss the advantages of this bound compared to the worldline one.

\subsection{The bound on the effective energy density}

This subsection is mainly based on results from \cite{Brown:2018hym}. The quantity of interest is the effective energy density (EED) given in Eq.~\eqref{eqn:eed}. The field equation for non-minimally coupled scalar fields is
\be \label{eqn:field}
(\Box_g+m^2+\xi R)\phi=0 \,,
\ee
where $\xi$ is the coupling constant and $R$ is the Ricci scalar. 
The constant $m$ has dimensions of inverse length, which would be the inverse Compton wavelength if one regarded~\eqref{eqn:field} as the starting-point for a quantum field theory with massive particles. The Lagrangian density is
\be
\label{eqn:lagrangian}
L[\phi]=\frac{1}{2} [(\nabla \phi)^2-(m^2+\xi R)\phi^2 ] \,,
\ee
where $(\nabla \phi)^2=-g^{\mu \nu} (\nabla_\mu \phi)(\nabla_\nu \phi)$. The stress energy tensor is obtained by varying the action with respect to the metric, giving
\be
\label{eqn:tmunu}
T_{\mu \nu}=(\nabla_\mu \phi)(\nabla_\nu \phi)-\frac{1}{2} g_{\mu \nu} (m^2 \phi^2-(\nabla \phi)^2)+\xi(-g_{\mu \nu} \Box_g-\nabla_\mu \nabla_\nu+G_{\mu \nu}) \phi^2 \,,
\ee
where $G_{\mu \nu}$ is the Einstein tensor. We should observe here that the field equation and the Lagrangian reduce to those of minimal coupling for flat spacetimes but the stress-energy tensor does not. 

The trace of the stress-energy tensor is given by
\be
\label{eqn:trace}
T = \left(\frac{n}{2}-1\right) (\nabla\phi)^2 - \frac{n}{2}m^2\phi^2 + \xi\left(-(n-1)\Box_g +  \left(1-\frac{n}{2}\right)R \right)\phi^2 \,.
\ee
We can write the stress energy tensor without mass dependence as
\bea
\label{eqn:tmbox}
T_{\mu \nu}&=&(1-2\xi)(\nabla_\mu \phi)(\nabla_\nu \phi)+\frac{1}{2} (1-4\xi)( \phi \Box_g \phi+(\nabla \phi)^2)g_{\mu \nu}-2\xi \phi \nabla_\nu \nabla_\mu \phi \nonumber\\
&& \qquad \qquad \qquad\qquad \qquad+\xi R_{\mu \nu} \phi^2 -\frac{1}{2} g_{\mu \nu} (\phi P_\xi \phi)\,,
\eea
where $P_\xi=\Box_g+m^2+\xi R$ is the Klein-Gordon operator. We also used the identity
\be
\label{eqn:boxidentity}
\phi \Box_g \phi=\frac{1}{2} \Box_g \phi^2-(\nabla\phi)^2 \,.
\ee
Alternatively we can remove the dependence on the D'Alembertian operator 
\bea
\label{eqn:tmmass}
T_{\mu \nu}&=&(1-2\xi)(\nabla_\mu \phi)(\nabla_\nu \phi)-\frac{1}{2} (1-4\xi)( m^2\phi^2+\xi R \phi^2-(\nabla \phi)^2)g_{\mu \nu}-2\xi \phi \nabla_\nu \nabla_\mu \phi \nonumber\\
&& \qquad \qquad \qquad\qquad \qquad+\xi R_{\mu \nu} \phi^2 -2\xi g_{\mu \nu} (\phi P_\xi \phi)\,.
\eea
Similarly for the trace, without explicit mass dependence we have
\be
\label{eqn:trbox}
T=\left(\frac{n}{2}-1-2\xi(n-1)\right) (\nabla\phi)^2+\left(\frac{n}{2}-2\xi(n-1)\right) \phi \Box_g \phi+\xi R\phi^2-\frac{n}{2} \phi P_\xi \phi \,,
\ee
while without dependence on the D'Alembertian operator
\bea
\label{eqn:trmass}
T&=&\left(\frac{n}{2}-1-2\xi(n-1)\right) (\nabla\phi)^2+\left(-\frac{n}{2}+2\xi(n-1)\right) m^2\phi^2 +\xi\left( 2\xi(n-1)+1-\frac{n}{2}\right) R\phi^2 \nonumber \\
&& \qquad \qquad -2\xi(n-1) \phi P_\xi \phi \,.
\eea
The EED without mass dependence is obtained by combining Eqs.~\eqref{eqn:eed}, \eqref{eqn:tmbox} and \eqref{eqn:trbox}
\bea
\label{eqn:rhobox}
\rho_U&=&(1-2\xi)U^\mu U^\nu (\nabla_\mu \phi) (\nabla_\nu \phi) +\frac{1-2\xi}{n-2} \phi \Box_g\phi-\frac{2\xi}{n-2} (\nabla \phi)^2 -2\xi \phi U^\mu U^\nu \nabla_\mu \nabla_\nu \phi \nonumber \\
&&+\xi R_{\mu \nu} U^\mu U^\nu \phi^2+\frac{\xi}{n-2}R\phi^2-\frac{1}{n-2}\phi P_\xi \phi \,,
\eea
where we assumed $U^\mu$ is a unit timelike vector. We get the EED without dependence on the D'Alembertian by combining Eqs.~\eqref{eqn:eed}, \eqref{eqn:tmmass} and \eqref{eqn:trmass}
\bea \label{eqn:rhomass}
\rho_U &=&  (1-2\xi) U^\mu U^\nu (\nabla_\mu \phi)(\nabla_\nu \phi)-\frac{1-2\xi}{n-2}  m^2 \phi^2 -\frac{2\xi}{n-2} (\nabla \phi)^2 -2\xi U^\mu U^\nu \phi \nabla_\mu \nabla_\nu  \phi  \nonumber\\
&&\quad +\xi U^\mu U^\nu R_{\mu \nu} \phi^2+\frac{2\xi^2}{n-2}  R \phi^2-\frac{2\xi}{n-2} (\phi P_\xi \phi)    \,.
\eea
The last term can be discarded ``on shell''  i.e. for $\phi$ satisfying  Eq.~(\ref{eqn:field}). For $\xi=0$ the EED further reduces to
\be
\rho_U=U^\mu U^\nu (\nabla_\mu \phi)(\nabla_\nu \phi)+\frac{1}{n-2} \phi \Box_g \phi \,,
\ee
or
\be \label{eqn:minSED}
\rho_U= U^\mu U^\nu (\nabla_\mu \phi)(\nabla_\nu \phi)-\frac{1}{n-2} m^2 \phi^2  \,.
\ee

We want to bound the following quantity where $x$ is a spacetime variable, $F(x)$ is a smearing function with compact support and the integral is over a spacetime volume
\be
\int dvol_g F^2(x) \rho_U \,.
\ee
In \cite{Brown:2018hym} the authors derived two bounds for this quantity from Eqs.~\eqref{eqn:rhobox} and \eqref{eqn:rhomass}.
Unlike the worldline case, for worldvolume averaging we can use successive integration-by-parts to derive a bound that has no explicit mass-dependence and remains free from any field derivatives.

Introducing $V^\mu = F(x)U^\mu$, the averaged EED for the nonminimally coupled scalar field ``on shell" obeys the following theorem \cite[Thm.\ 2]{Brown:2018hym}.
\begin{theorem}
	\label{the:clasvol}
	If $M$ is a manifold with metric $g$ and dimension $n \geq 3$, $T_{\mu \nu}$ the stress-energy tensor of a scalar field with coupling constant $\xi \in [0,\xi_v]$
	and $f$ a real valued function on $M$ with compact support, then ``on shell ''  
	\be
	\int dvol_g \, \rho_U \, F^2(x)  \geq -\min \{ \mathcal{B}_1,\mathcal{B}_2\} \,,
	\ee
	where
	\bea
	\mathcal{B}_1&=&\int dVol \bigg\{ \frac{1-2\xi}{n-2}  m^2 F^2(x) -\frac{2\xi^2 R}{n-2} F^2(x) \nonumber\\
	&&\qquad \qquad \qquad \qquad+ \xi  \left[(\nabla_\mu V^\mu)^2+(\nabla_\mu V^\nu)(\nabla_\nu V^\mu)\right] \bigg\}  \phi^2 \,,
	\eea
	and 
	\bea
	\label{eqn:b2curve}
	\mathcal{B}_2&=&\int dvol_g \bigg\{ \frac{1-2\xi}{2(n-2)}(\Box_g F^2(x)) - \frac{\xi R}{n-2} F^2(x) \nonumber\\
	&&\qquad \qquad \qquad \qquad+\xi  [(\nabla_\mu V^\mu)^2+(\nabla_\mu V^\nu)(\nabla_\nu V^\mu)] \bigg\} \phi^2 \,.
	\eea
\end{theorem}
Here
\be
\xi_v= \frac{n-3}{2(n-2)} \,,
\ee 
where $\xi_v < 2\xi_c$ for any spacetime dimension $n>2$, while $\xi_c<\xi_v$ for $n\ge 4$. By $\xi_c$ we denote the conformal coupling in $n$ dimensions
\be
\xi_c=\frac{n-2}{4(n-1)} \,.
\ee
For minimally coupled fields on any spacetime, 
\be
\mathcal{B}_1=\frac{1}{n-2} m^2  \int dvol_g \, F^2(x)   \phi^2 \,, \text{ and } \mathcal{B}_2= \frac{1}{2(n-2)} \int dvol_g \, (\Box_g F^2(x)) \phi^2 \,.
\ee
For flat spacetimes the bounds of Theorem \ref{the:clasvol} become
\be \label{eqn:B1flat}
\mathcal{B}_1=\int dvol_g \left\{ \frac{1-2\xi}{n-2}  m^2 F^2(x) +\xi  [(\nabla_\mu V^\mu)^2+(\nabla_\mu V^\nu)(\nabla_\nu V^\mu)] \right\}  \phi^2 \,,
\ee
and
\be \label{eqn:B2flat}
\mathcal{B}_2=\int dvol_g \bigg\{ \frac{1-2\xi}{2(n-2)} (\Box F^2(x)) + \xi  [(\nabla_\mu V^\mu)^2+(\nabla_\mu V^\nu)(\nabla_\nu V^\mu)] \bigg\} \phi^2 \,.
\ee
In comparison, the worldline bound derived in \cite{Brown:2018hym} is
\bea
\label{eqn:cline}
\int_\gamma dt \, \rho_U\, f^2(t)&\geq& - \int_\gamma dt \bigg\{ \frac{1-2\xi}{n-2} m^2 f^2(t)+\xi \bigg(2( f'(t))^2+R_{\mu \nu} \dot{\gamma}^\mu \dot{\gamma}^\nu f^2(t) \nonumber\\
&& \qquad \qquad  - \frac{2\xi}{n-2} Rf^2(t) \bigg) \bigg\}\phi^2  \,,
\eea
where the EED is averaged over a segment of a timelike geodesic $\gamma$. We note that only in the case of worldvolume averaging can we have a bound without mass dependence and field derivatives. That kind of bound has another important advantage which we can see by rescaling the support of the smearing function. We first let $\fmax$ be the maximum amplitude of the field 
\be
\label{eqn:fmax}
\fmax=\sup_M |\phi | \,,
\ee
so we can take it out of the bound. Then we look at the bound $\mathcal{B}_2$ in flat spacetimes \footnote{This bound differs from the analogous one in \cite{Brown:2018hym} where the absolute values were incorrectly omitted.}
\bea
\int dvol_g \, \rho_U\, F^2(x) &\geq& - \fmax^2 \int dvol_g \bigg\{ \frac{1-2\xi}{2(n-2)}  (\left| \Box_g F^2(x) \right|) \nonumber\\
&&\qquad \qquad \qquad \qquad+ \xi  [(\nabla_\mu V^\mu)^2+\left|(\nabla_\mu V^\nu)(\nabla_\nu V^\mu)\right|] \bigg\} \,.
\eea
Next, we consider a translationally invariant unit timelike vector field $U^\mu$ and define the rescaled smearing function $F_\lambda$ for $\lambda>0$ to be
\be
F_\lambda(x)=\frac{F(x/\lambda)}{\lambda^{n/2}} \,,
\ee
so that its normalization is independent of the choice of $\lambda$. Then in the limit of large $\lambda$ the bound goes to zero and we have
\be
\liminf_{\lambda \to \infty} \int dvol_g \, \rho_U\, F^2_\lambda(x) \geq 0 \,,
\ee
thus establishing an averaged SEC (ASEC) for flat spacetimes. A similar calculation for the $\mathcal{B}_1$ bound gives a weaker, negative, bound in this case. The ASEC cannot be proven with the same assumptions for the worldline bound. In particular for 
\be
\fmax=\sup_\gamma |\phi | \,,
\ee
and
\be
f_{\lambda}(t)=\frac{f(t/\lambda)}{\sqrt{\lambda}} \,,
\ee
we get
\be
\liminf_{\lambda \to \infty} \int_\gamma dt \, \rho_U \, f^2_\lambda (t)  \geq -  \frac{1-2\xi}{n-2} m^2 \fmax^2 \,,
\ee
as shown in \cite{Brown:2018hym}. So for the wordline bound the long time average can be negative even for minimal coupling.

\subsection{The curvature bound}

To derive a bound for the singularity theorem, we need to apply the Einstein equation
\be
G_{\mu \nu}=8\pi T_{\mu \nu} \,.
\ee
Taking the trace of the Einstein equation $G_{\mu\nu}=8\pi T_{\mu\nu}$ gives
\begin{equation}\label{eq:trEins}
	\left(-\frac{n}{2}+1\right)R = 8\pi T \,.
\end{equation}
Using Eq.~\eqref{eqn:trmass} ``on shell'' and using a similar method as \cite{Brown:2018hym} we can show
\bea\label{eqn:trace_onshell}
T &=& 2(n-1)(\xi_c-\xi)  (\nabla\phi)^2 - \left(1+2(n-1)(\xi_c-\xi) \right)m^2\phi^2  \nonumber\\
&& -2(n-1)(\xi_c-\xi)\xi R \phi^2 \,.
\eea
Using Eqs.~\eqref{eq:trEins} and \eqref{eqn:trace_onshell} we have 
\begin{equation}
	\label{eqn:teqn}
	\frac{n-2}{16\pi} \left(1-8\pi\xi(1-\xi/\xi_c)\phi^2\right)R = -2(n-1)(\xi_c-\xi)  (\nabla\phi)^2 + \left(1+2(n-1)(\xi_c-\xi) \right)m^2\phi^2  \,.
\end{equation}
Next, we assume that $8\pi \xi \phi^2 \leq 1/2$. Physically, it makes sense to consider values of $\phi^2$ such that $8\pi \xi \phi^2 \ll 1$ \cite{Fliss:2023rzi}. Considering $8\pi \xi \phi^2$ close to $1$, means considering field values close to the Planck scale, when our approach breaks down. Under that assumption we have $1-8\pi\xi(1-\xi/\xi_c)\phi^2\geq 8\pi\xi\phi^2$. Using this and the value of $\xi_c$ equation \eqref{eqn:teqn} yields
\be
R\geq - \frac{ (\xi_c-\xi)}{\xi_c \xi \phi^2} (\nabla \phi)^2 \,.
\ee
Now we can replace the term including the Ricci scalar in the expression for EED of Eq.~\eqref{eqn:rhobox} `on shell'
\bea
\rho_U & \geq & (1-2\xi)U^\mu U^\nu (\nabla_\mu \phi) (\nabla_\nu \phi) +\frac{1-2\xi}{n-2} \phi \Box_g\phi-\frac{(\xi_c(1+2\xi)-\xi)}{\xi_c(n-2)} (\nabla \phi)^2 -2\xi \phi U^\mu U^\nu \nabla_\mu \nabla_\nu \phi \nonumber \\
&&+\xi R_{\mu \nu} U^\mu U^\nu \phi^2\,,
\eea
Using the Einstein equation and the definition of EED
\be
8\pi \rho_U= R_{\mu \nu} U^\mu U^\nu \,,
\ee
we have the purely geometric bound 	
\bea
R_{\mu \nu}U^\mu U^\nu (1-8\pi \xi \phi^2) & \geq & 8\pi \bigg((1-2\xi)U^\mu U^\nu (\nabla_\mu \phi) (\nabla_\nu \phi) +\frac{1-2\xi}{n-2} \phi \Box_g\phi  \nonumber \\
&&-\frac{(\xi_c(1+2\xi)-\xi)}{\xi_c(n-2)} (\nabla \phi)^2-2\xi \phi U^\mu U^\nu \nabla_\mu \nabla_\nu \phi\bigg)\,.
\eea
Using the identity of Eq.~\eqref{eqn:boxidentity} we have
\bea
R_{\mu \nu}U^\mu U^\nu (1-8\pi \xi \phi^2) & \geq & 8\pi \bigg((1-2\xi)U^\mu U^\nu (\nabla_\mu \phi) (\nabla_\nu \phi) -\frac{(2\xi_c-\xi)}{\xi_c(n-2)} (\nabla \phi)^2  \nonumber \\
&&+\frac{1-2\xi}{2(n-2)} \Box_g\phi^2-2\xi \phi U^\mu U^\nu \nabla_\mu \nabla_\nu \phi\bigg)\,.
\eea
Combining the first two terms gives
\bea
\label{eqn:comb}
(1-2\xi)U^\mu U^\nu (\nabla_\mu \phi) (\nabla_\nu \phi) -\frac{(2\xi_c-\xi)}{\xi_c(n-2)} (\nabla \phi)^2&=& k(\xi,n) U^\mu U^\nu (\nabla_\mu \phi)(\nabla_\nu \phi) \nonumber\\
&&+\frac{(2\xi_c-\xi)}{\xi_c(n-2)}h^{\mu\nu}(\nabla_\mu \phi)(\nabla_\nu \phi) \,,
\eea
where $h^{\mu \nu}=g^{\mu \nu}+U^\mu U^\nu$ is a positive semi-definite metric and 
\be
k(\xi,n)=\frac{ \xi_c((n-4)-2\xi(n-2))+\xi }{\xi_c(n-2)} \,.
\ee
Both terms in Eq.~\eqref{eqn:comb} are positive for $n>3$ as $k(\xi,4)=\xi$ and $k(\xi,n)>0$ for $n>4$ and $0\leq\xi \leq \xi_c$. Then, using the smearing function $F(x)$ we can write
\be
\int dvol_g F(x)^2 R_{\mu \nu}U^\mu U^\nu (1-8\pi \xi \phi^2)  \geq  -8\pi \int dvol_g F(x)^2 \bigg(-\frac{1-2\xi}{2(n-2)} \Box_g\phi^2+2\xi \phi U^\mu U^\nu \nabla_\mu \nabla_\nu \phi\bigg)\,.
\ee
Integrating by parts the second term gives
\bea
&& 2\xi \int dvol_g F(x)^2 \phi U^\mu U^\nu \nabla_\mu \nabla_\nu \phi=-2\xi \int dvol_g F(x)^2(\nabla_U \phi)^2-2\xi \int dvol_g \nabla_U (F(x)^2)\phi \nabla_U \phi \nonumber\\
&& \qquad  -2\xi \int dvol_g F(x)^2 \phi A^\nu \nabla_\nu \phi-2\xi \int dvol_g F(x)^2 \theta \phi (\nabla_U \phi) \,,
\eea
where we defined the expansion $\theta=\nabla_\mu U^\mu$ (note that this is $\mathrm{tr}(\nabla U)$ in index-free notation). Here $\nabla_U U^\mu=A^\mu$ is the acceleration.

From now on we will consider the geodesic flow emanating normally from the Cauchy surface $\Sigma$ as discussed in previous sections. Then the $U^\mu$ is tangent to the timelike geodesic parametrized by proper time $t$, $A=0$ and we have
\bea
\label{eqn:boundint}
&&\int dvol_g F(x)^2 R_{\mu \nu}U^\mu U^\nu (1-8\pi \xi \phi^2)  \geq \nonumber\\
&& \qquad \qquad \qquad-8\pi \int dvol_g  \bigg(\frac{1-2\xi}{2(n-2)}\phi^2 \Box_g F(x)^2-2\xi \frac{\partial (F(x)^2)}{\partial t} \phi \dot{\phi}-2\xi F(x)^2 \theta \phi\dot{\phi} \bigg)\,.
\eea
From Remark~\ref{rem: dAlembertop} we have
\be
\Box_g F^2(x)= \frac{1}{\sqrt{\det g_t}} \partial_t(\sqrt{\det g_t} \partial_t
(F(x)^2)) -\Delta_{g_t} F^2(x) \,.
\ee
Replacing $\Box_g F^2(x)$ in Eq.~\eqref{eqn:boundint}, integrating by parts and noting that $dvol_g=\sqrt{\det g_t} dtd\sigma(\x) $ gives
\bea
\label{eqn:boundint2}
&&\int dvol_g F(x)^2 R_{\mu \nu}U^\mu U^\nu (1-8\pi \xi \phi^2)  \geq -8\pi \int dvol_g  \bigg(\frac{1-2\xi}{2(n-2)}F(x)^2 (-\Delta_{g_t} \phi^2) \nonumber\\
&& \qquad \qquad -\frac{1}{(n-2)} 2 F(x) \dot{F}(x) \phi\,\dot{\phi}-2\xi F(x)^2 \theta \phi \dot{\phi} \bigg) \,.
\eea
 Since $|\phi|<(8\pi \xi)^{1/2}$ we can absorb the factor $(1-8\pi \xi \phi^2)^{1/2}$ in $F(x)$
\bea
&&\int dvol_g F(x)^2 R_{\mu \nu}U^\mu U^\nu \geq -8\pi \int dvol_g  \bigg(\frac{1-2\xi}{2(n-2)}\frac{F(x)^2}{ (1-8\pi \xi \phi^2)} (-\Delta_{g_t} \phi^2) \nonumber\\
&& \qquad \qquad -\frac{2}{(n-2)}  \frac{F(x) \dot{F}(x)}{ (1-8\pi \xi \phi^2) } (\phi\,\dot{\phi})-2\xi \frac{F(x)^2}{(1-8\pi \xi \phi^2)} \theta \phi \dot{\phi} \bigg) \,,
\eea
where we omitted a strictly positive term. Using the following definitions
\be
\label{eqn:fmaxprime}
\dot{\phi}_{\max}=\sup_M \left|\frac{\partial}{\partial t}\phi \right| \,,
\quad \Delta \fmax^2=\sup \left|\Delta_{g_t} \phi^2 \right|
\ee
along with Eq.~\eqref{eqn:fmax}, we have
\bea
\label{eqn:riccwtheta}
&&\int dvol_g F(x)^2 R_{\mu \nu}U^\mu U^\nu \geq -8\pi \int dvol_g  \bigg\{ \frac{F(x)^2}{ (1-8\pi \xi \fmax^2)}  \bigg(\frac{1-2\xi}{2(n-2)}\Delta \fmax^2 \nonumber\\
&& \qquad  + \frac{\fmax\,\dot{\phi}_{\max}}{T(n-2)}+2\xi |\theta| \fmax \dot{\phi}_{\max}\bigg) +\frac{\dot{F}(x)^2}{ (1-8\pi \xi \fmax^2) } \frac{T \fmax\,\dot{\phi}_{\max}}{(n-2)}  \bigg\} \,.
\eea
Here we also used the inequality of Eq.~\eqref{eqn:elineq}.

\subsection{A singularity theorem for the Einstein-Klein-Gordon theory}

To proceed we need to bound the absolute value of the expansion $\theta$. First we note that the singularity theorems \ref{theorem: m=1} or \ref{rem: m>1case} use a function $F \in C^{\infty}_c(M)$ with support entirely contained in $\Omega_T^+(\Regplus(T))$ such that
\begin{align}
    (F \circ \exp_{\Sigma})(t, \x) = f(t) h(\x) \,,
\end{align}
with $f \in C^{\infty}_c(\R)$ with support contained in $[0,T]$ and $h \in C^{\infty}_c(\Sigma)$ with support contained in $\Regplus(T)$. 

Now since $\theta(t,\x)=H(\exp_{\x}(t \vec{n}_{\x}))$ on $\Omega_T^+(\Regplus(T))$, we can use Lemma \ref{lem: backwarddalembertcomparison} to obtain
\begin{align}
    \label{eq:theatboundbelow}
\theta(t,\x) \geq -(n-1)\sqrt{|\kappa|} \coth(\eta \sqrt{|\kappa|})=C^{\Box-} \;\;\;\mathrm{on}\;\Omega_T^+(\Regplus(T))\cong [0,T]\times \Regplus(T).
\end{align}
For an upper bound we may use Lemma s \ref{lem: forwarddalembertcomparison}  together with Remark \ref{rem: comparisonconstants} (as we assume $H\leq \beta <0$ and we can always assume $ \beta >-(n-1)\sqrt{|\kappa|}$ by making it smaller if necessary), i.e., standard mean curvature comparison, to obtain
\begin{align}
    \label{eq:theatboundabove}
\theta(t,\x) &=H((\exp_{\x}(t \vec{n}_{\x}))) \leq H_{\kappa, \beta }(T)= C^{\Box+} \;\;\;\mathrm{on}\;\Omega_T^+(\Regplus(T))\cong [0,T]\times \Regplus(T).
\end{align}

So, as in Remark \ref{rem: dAlembertop},
\begin{align}
    \label{eq:theatbound}
|\theta(t,\x)| &\leq \max\{C^{\Box-},|C^{\Box+}|\} = (n-1)\sqrt{|\kappa|} \coth(\eta \sqrt{|\kappa|}).
\end{align} 

From this \eqref{eqn:riccwtheta} becomes 
	\bea
	\label{eqn:riccithree}
	&&\int dVol \, R_{\mu \nu} U^\mu U^\nu \, F^2(x)  \geq - Q_1  ||F||^2 -Q_2\| \dot{F}\|^2 \,,
	\eea
for any $F \in C^{\infty}_c(M)$ with support entirely contained in $\Omega_T^+(\Regplus(T))$, where the constants are
\be\label{eqn:q1}
Q_1= \frac{8\pi}{1-8\pi \xi \fmax^2} \bigg[\frac{1-2\xi}{2(n-2)} \Delta \fmax+\frac{\fmax \dot{\phi}_{\max}}{T(n-2)}+2\xi(n-1)\sqrt{|\kappa|} \coth{(\eta \sqrt{|\kappa|})} \fmax \dot{\phi}_{\max} \bigg] \,,
\ee
and
\be
\label{eqn:q2}
Q_2=\frac{8\pi}{1-8\pi \xi \fmax^2}  \frac{1}{n-2} T\fmax \dot{\phi}_{\max} \,.
\ee
We note that this bound is mass independent unlike the worldline bound derived in Ref.~\cite{Brown:2018hym}. The appearance of spatial derivatives of the field can be considered small or negligible in some cases, for example those of cosmological spacetimes. We note that in geometric units the constant $Q_1$ has dimensions of inverse time square while the constant $Q_2$ is dimensionless.

We are ready to state a singularity theorem for the non-minimally coupled classical scalar field.  

\begin{cor}
Let $(M, g, \phi)$ be a solution to the Einstein–Klein–Gordon equation in $n > 2$ spacetime  dimensions. Let $(M, g, \phi)$ be globally hyperbolic with a smooth spacelike Cauchy surface  $\Sigma$. Suppose $\phi$ has coupling $\xi \in [0, 
\xi_c]$. Assume hypotheses (i) and (iii) from Theorem \ref{theorem: m=1} 
with constants $Q_1$ and $Q_2$ given by Eqs.~\eqref{eqn:q1} and 
\eqref{eqn:q2}. The constant $\fmax$ is defined in \eqref{eqn:fmax} and obeys the inequality $8\pi \xi \fmax^2 \leq 1/2$. The constants $\dot{\phi}_{\max}$ and $\Delta \fmax^2$ are defined as in \eqref{eqn:fmaxprime}. 
Then there exists a $\tau$ such that no future-directed timelike curve emanating from $\Sigma$ has length greater than $\tau$ and $(M,g)$ is future timelike geodesically incomplete.
\end{cor}
\begin{proof}
First we note that $(M,g,\phi)$ obeys the worldvolume bound of Eq.~\eqref{eqn:riccithree} for any $F \in C^{\infty}_c(M)$ with support entirely contained in $\Omega_T^+(\Regplus(T))$ (which is sufficient by Remark \ref{rem: m>1case}(1)). Thus the condition (ii) of Theorem \ref{theorem: m=1} is also obeyed with constants $Q_1$ and $Q_2$ given by Eqs.~\eqref{eqn:q1} and \eqref{eqn:q2}. The result follows from Theorem \ref{theorem: m=1}.
\end{proof}

\section{Application}
\label{section: application}

In this section we estimate the required extrinsic curvature to have past timelike geodesic incompleteness for the non-minimally coupled scalar field in a simple cosmological model. The approach is based on the analysis performed in \cite{Brown:2018hym} and \cite{Fewster:2019bjg}. The results will be useful to compare our theorem with the worldline ones proven in the previous references.

One important estimate that needs to be made is the maximum allowed field value $\fmax$. To do that we follow Ref.~\cite{Brown:2018hym} and estimate it using a hybrid model. The idea is to estimate the (classical) $\fmax$ using the Wick square of a KMS state in Minkowski spacetime. That way we can connect the field value with a temperature $T$. 

Assuming the scalar field describes an elementary particle with mass $m$, a maximum temperature where the model can be trusted is taken as much smaller than the Compton temperature $T_m=mc^2/k$ where $k$ is the Boltzmann constant. Then $T_{\max}=10^{-2}T_m \ll T_m$. With these considerations we have \cite{Brown:2018hym}
\be
\fmax^2\sim \langle \nord{\phi^2} \rangle_{T_{\max}}\sim 5\times 10^{-4} \frac{m^2 c^3}{\bar{h}} K_1(100) \,,
\ee
where $K_x$ is the modified Bessel function. 

On dimensional grounds the ratio $\dot{\phi}_{\max}/\fmax$ is proportional to $c/\lambda$ where $\lambda$ is taken as the Compton length of elementary particles. Then
\be
\dot{\phi}_{\max}=\fmax \frac{mc^2}{h} \,.
\ee
Since we will consider a cosmological setting we assume the field is spatially homogeneous and so $\Delta \fmax^2=0$. The values of the field and its derivative only depend on the mass of an (elementary) particle. To compare our results with Refs.~\cite{Brown:2018hym} and \cite{Fewster:2019bjg} we use the pion, the proton and the Higgs. The pion has mass $2.40 \times 10^{-28}$kg, the proton $1.67 \times 10^{-27}$kg and the Higgs $2.24 \times 10^{-25}$ kg.

Since we have been using geometric units, we need to restore the dimensions for this physical application. The quantity $C^{\square}_{\max}$ becomes
\be
 C^{\square}_{\max} = \max\{C^{\Box-},|C^{\Box+}|\} = 3\frac{\sqrt{G|\kappa|}}{c}\coth(\eta \frac{\sqrt{G|\kappa|}}{c}) \,,
\ee
so it has dimensions of inverse time. Similarly we restore dimensions in the constants $Q_1$ and $Q_2$. Both constants are multiplied by $G/c^4$, meaning $Q_2$ is dimensionless and $Q_1$ has dimensions of inverse time square. 

After restoring the dimensions we proceed with an estimation and optimization of the required extrinsic curvature $\nu_{\tau_0,\eta}$ of Theorem~\ref{theorem: m=1}. First we optimize in terms of $\tau_0$ with the condition $T \gg \tau_0$. Then the optimal $\tau_0$, meaning the one that gives the minimum of $\nu_{\tau_0,\eta}$ is
\be
\tau_0^{\text{opt}}=\sqrt{\frac{3}{2}}c\sqrt{\frac{Q_2(1+TC^{\square}_{\max}/2)}{G|\kappa|}} \,.
\ee

To proceed we will make some further approximations. First, for the rest of the calculation we assume $\xi=1/6$, the conformal coupling in $4$-dimensions. In this case any coupling in the allowed range is expected to give similar results. Then we assume that 
\be
y \equiv \left(\eta \frac{\sqrt{G|\kappa|}}{c}\right) \ll 1 \,.
\ee
With this assumption we can take $\cosh{y} \approx 1$ and $\coth{y} \approx y^{-1}$. The maximum allowed value of $\eta$ so that the approximation is valid is $10^{-2} c/\sqrt{G\kappa}$. We fix the value of $\eta$ to be equal to that and we will revisit it to make sure $\eta \ll T$ as expected.

With these assumptions, for each particle, the required extrinsic curvature is a function of $T$ and $\kappa$. As we expect $\eta \ll T$ we will refer to $T$ as the maximum time to singularity. To compare the results of our theorem to cosmological data we will use results from the PLANCK collaboration \cite{Planck:2018vyg}.  The Hubble parameter is the same as the extrinsic curvature of a cosmological spacelike Cauchy surface and today it has an estimated value of $H_0=2.18 \times 10^{-18} s^{-1}$ while the age of the universe is $t_0=4.35 \times 10^{17}s$ \footnote{There is an experimental ambiguity regarding the value of the Hubble constant. Here we use the value measured using the CMB perturbations.}. The goal is to find the largest possible $\kappa$ so that $\nu \leq H_0$ and $T$ is as close as possible to $t_0$. 

We note that $\nu$ has a term of the form $\sim 3/T$ similar to the worldline theorem of Ref.~\cite{Fewster:2019bjg}. That means $T$ has to be larger than $1.5 \times 10^{18} s$ to have the required extrinsic curvature smaller than the measured one. Considering that we want a prediction of the singularity time we examine values of $T$ between this value and $10^{20} s$, about three orders of magnitude larger than $t_0$. For each elementary particle and within this range of $T$ we found the maximum possible value of $\kappa$ for our theorem to apply. The results are presented in Tab.~\ref{tab:particles}. 

\begin{table}
	\begin{center}\begin{tabular}{|c|c|c|c|c|} \hline
			Particle & Mass in kg &  $T$ in $s$ & $\kappa$ in $J/m^3$ & $\eta$ in $s$  \\ \hline \hline
			Pion & $2\times 10^{-28}$  & $3.05 \times 10^{18}$ & $5.18 \times 10^{20}$ & $16.18$  \\ \hline
			Proton & $1.67 \times 10^{-27}$  & $2.94 \times 10^{18}$ & $1.07 \times 10^{19}$ & $112.52$  \\ \hline
			Higgs & $2.23 \times 10^{-25}$  & $3.06 \times 10^{18}$ & $5.93 \times 10^{14}$ & $1.51 \times 10^{4}$ \\ \hline
		\end{tabular}\caption{The masses, maximum time to singularity and required $\kappa$ to have past geodesic incompleteness for three elementary particles. We note that in all cases $\eta \ll T$.}
		\label{tab:particles}
	\end{center}
\end{table} 

Comparing our results with the ones from Ref.~\cite{Fewster:2019bjg} where the same model was examined using a worldline inequality we find several differences. First of all with the worldvolume inequality we are able to predict past geodesic incompleteness for the Higgs which was impossible in the worldline case. Second, our method doesn't require the SEC to hold near the Cauchy surface, thus this assumption was dropped and we are able to apply the theorem using the spacelike cosmological Cauchy surface today. The drawback of our theorem is the requirement for a global pointwise condition. However, we find that the negative energy density allowed at each spacetime point is very large, comparable to the energy density at the very early universe.

\section{Discussion}
\label{sec:discussion}

In this work we prove the first singularity theorems using a worldvolume inequality as an assumption. Using area comparison results in a Lorentzian manifold, we are able to derive bounds on individual timelike geodesics from worldvolume integrals. Additionally, using index form methods developed in \cite{Fewster:2019bjg} we prove a singularity theorem with a worldvolume integrated energy condition. Our results have physical applications, and as an example we apply our theorem in the case of the classical non-minimally coupled field, which violates the strong energy condition. Finally, we apply our theorem using a toy cosmological model and compare our results with the ones from worldline inequalities. We find that the worldvolume theorem could give better results especially in the case of elementary particles with large mass such as the Higgs boson. 

The obvious extension of this work is to apply the theorem to a quantum field theory using a quantum energy inequality as an assumption as done in Ref.~\cite{Fewster:2021mmz}. There are currently a few obstacles for this application: Theorem \ref{theorem: m=1} was proven only for one timelike derivative on the smearing function while the quantum inequality requires two. Perhaps more importantly, the theorem has as an assumption a pointwise energy condition. Even though the negative energy allowed can be large, in the context of quantum field theory there are no lower energy bounds for spacetime points. Thus the pointwise condition would need to be replaced with an integral one.

The null geodesic incompleteness case seems even more interesting than the timelike one. As it is impossible to bound the null energy on single geodesic segments \cite{Fewster:2002ne} but it is possible on finite null surfaces \cite{Fliss:2021phs} , it would be of great interest to prove a Penrose-type theorem for worldvolume inequalities. One mathematical difficulty arising in the null case is that there are no direct Riemannian analogues to draw from, however there are null area comparison results, cf. \cite{choquet2009light, grant2011areas}, controlling the $(n-2)$-dimensional area of constant-affine-parameter slices of the lightcone. This work is currently in progress.

\begin{acknowledgments}
E-AK is supported by the ERC Consolidator Grant QUANTIVIOL. This work is part of the $\Delta$-ITP consortium, a program of the NWO that is funded by the Dutch Ministry of Education, Culture and Science (OCW). AO acknowledges the support by the project P33594 of the Austrian Science Fund FWF as well as the ÖAW-DOC scholarship of the Austrian Academy of Sciences. MG is supported by the German research foundation (DFG) focus program SPP2026 {\em Geometry at Infinity} and YS was supported by the Uni:Docs program of the University of Vienna.
The authors would like to thank Chris Fewster, Michael Kunzinger and Roland Steinbauer for helpful discussions. They are also grateful to the Mathematical Research Institute of Oberwolfach (MFO) for hosting them for two weeks in April 2022 during the course of the ``Research in Pairs" program, which contributed to this project.
\end{acknowledgments}

\newpage

\bibliographystyle{utphyssort}
\bibliography{biblio}

\end{document}